\documentclass{article}
\usepackage[margin=1in]{geometry} 
\usepackage[utf8]{inputenc}

\usepackage{verbatim}
\usepackage{subcaption}
\usepackage{amsmath}
\usepackage{amssymb}
\usepackage{amsthm}
\usepackage[dvipsnames]{xcolor}
\usepackage{algorithm}
\usepackage{algorithmic}
\usepackage{float}
\usepackage{soul}
\usepackage{graphicx,hyperref}
\usepackage{hyperref}
\usepackage{enumitem}
\usepackage{multicol}
\usepackage[labelfont=bf]{caption}

\newtheorem{theorem}{Theorem}[section]
\newtheorem{lemma}[theorem]{Lemma}
\newtheorem{corollary}[theorem]{Corollary}

\newcommand{\im}{\text{im}\,}
\newcommand{\boundary}{\partial}

\newcommand{\Z}{\mathbb{Z}}
\newcommand{\R}{\mathbb{R}}
\newcommand{\suchthat}{\,|\,}

\newcommand{\V}{\mathbf{V}} 
\DeclareMathOperator{\tw}{tw}

\usepackage{environ}

\begin{document}

\author{
Mitchell Black
    \thanks{School of Electrical Engineering and Computer Science, Oregon State University;\href{
            mailto:blackmit@eecs.oregonstate.edu
        }{
            blackmit@eecs.oregonstate.edu
        }.
    }
\and
Amir Nayyeri
    \thanks{
        School of Electrical Engineering and Computer Science, Oregon State University;\href{
            mailto:nayyeria@eecs.oregonstate.edu
        }{
            nayyeria@eecs.oregonstate.edu
        }.
    }
}
\title{
Finding minimum bounded and homologous chains in simplicial complexes with bounded-treewidth 1-skeleton
}
\date{}

\maketitle
\thispagestyle{empty}
\begin{abstract}
We consider two problems on simplicial complexes: the Optimal Bounded Chain Problem and the Optimal Homologous Chain Problem. The Optimal Bounded Chain Problem asks to find the minimum weight $d$-chain in a simplicial complex $K$ bounded by a given $(d{-}1)$-chain, if such a $d$-chain exists. The Optimal Homologous Chain problem asks to find the minimum weight $(d{-}1)$-chain in $K$ homologous to a given $(d{-}1)$-chain. Both of these problems are NP-hard and hard to approximate within any constant factor assuming the Unique Games Conjecture. We prove that these problems are fixed-parameter tractable with respect to the treewidth of the 1-skeleton of $K$. 
\end{abstract}

\newpage
\pagenumbering{arabic}

\section{Introduction}

We consider two problems in this paper: the Optimal Bounded Chain Problem and the Optimal Homologous Chain Problem. Both of these problems are NP-hard \cite{chen_freedman, dunfield-harani}, so it is natural to ask whether more efficient algorithms exist for any special family of simplicial complexes. In this paper, we study algorithms for these problems for simplicial complexes with bounded-width tree decompositions of their 1-skeletons.
\par
Tree decompositions have been used successfully to design algorithms on graphs; see \cite{Cygan_2015}. Often, graphs with tree decompositions of bounded-width admit polynomial-time solutions to otherwise hard problems. Recently, tree decompositions have also begun to be used for algorithms on simplicial complexes. Existing algorithms use tree decompositions of a variety of graphs associated with a simplicial complex. The most commonly used graph is the dual graph of combinatorial $d$-manifolds \cite{bagchi-tightness, burton-courcelle, Burton_morse, burton_taut}. Other graphs that have been used are the incidence graph between the $d$ and $(d+1)$-simplices (a.k.a.~level $d$ of the Hasse diagram) \cite{Burton_morse, blaser_hl, blaser_mbc}, the adjacency graph of the $d$-simplices \cite{blaser_hl}, and the 1-skeleton \cite{bagchi-tightness}. 
\par
Our algorithms use a tree decomposition of the 1-skeleton. We observe that any simplex of a simplicial complex will necessarily be contained in some bag of a tree decomposition of its 1-skeleton. This is analogous to the requirement that each vertex and edge of a graph be contained in some bag of a tree decomposition. In this way, a tree decomposition of the 1-skeleton of a simplicial complex is analogous to a tree decomposition of a graph.
\par
The problems we consider ask to find a chain in a simplicial complex with certain properties. There can be exponentially many such chains, so we cannot simply check whether each chain or subcomplex has the desired property. Our approach is to incrementally build these chains, checking along the way that they have the desired property. Tree decompositions of the 1-skeleton of a simplicial complex define a recursive structure on a simplicial complex that we use to incrementally build our solutions.

\subsection{Our Results}

We now summarize the main results of our paper. Throughout, $n$ is the number of vertices in the input simplicial complex and $k$ is the treewidth of the 1-skeleton of the input simplicial complex. Formal definitions of the terms used will be given in Section \ref{sec:prelim}.
\par
The \textit{\textbf{Optimal Bounded Chain Problem}} (OBCP) asks to find the minimum weight $d$-chain $c$ of a simplicial complex $K$ bounded by a given $(d{-}1)$-chain $b$, if such a chain $c$ exists. We work with chain groups with coefficients in $\Z_2$, so a $d$-chain can therefore be thought of as a set of $d$-simplices. The weight of a $d$-chain is the number of $d$-simplices it contains.  We show that OBCP is fixed-parameter tractable with respect to the treewidth of the $1$-skeleton.
\begin{theorem}
\label{thm:obc}
    OBCP can be solved in $2^{O\left(k^{d}\right)}n\subset 2^{O\left(k^{k}\right)}n$ time.
\end{theorem}

As a corollary, our algorithm for OBCP can determine whether or not two $(d{-}1)$-chains $b$ and $h$ are homologous. The naive algorithm for this problem is to solve the system of linear equation $\boundary c=b+h$ for $c$, which takes $O(n^{\omega})$ time.  Our algorithm runs in linear time on a bounded-treewidth simplicial complex.

\begin{corollary}
\label{cor:homology_test}
    We can determine whether two $(d{-}1)$-chains are homologous in $2^{O\left(k^{d}\right)}n\subset 2^{O\left(k^{k}\right)}n$ time.
\end{corollary}

The \textit{\textbf{Optimal Homologous Chain Problem}} (OHCP) asks to find the minimal weight $(d{-}1)$-chain $h$ homologous to a given $(d{-}1)$-chain $b$. If the chains $b$ and $h$ are cycles, OHCP is also known as Homology Localization.
We show that OHCP is fixed-parameter tractable algorithm with respect to the treewidth of the $1$-skeleton.
\begin{theorem}
\label{thm:ohc}
    OHCP can be solved in $2^{O\left(k^{d}\right)}n\subset 2^{O\left(k^{k}\right)}n$ time.
\end{theorem}

Our results for OBCP and OHCP easily generalize to weighted complexes. We discuss this at the end of Section \ref{section:mbc}. 

\section{Related Work}

Blaser and V\.agset~\cite{blaser_hl} and Blaser et al.~\cite{blaser_mbc} independently discovered treewidth-parameterized algorithms for OHCP and OBCP concurrently to this paper. Their algorithms uses different graphs of the simplicial complex: the adjacency graph of the $d$-simplices and the incidence graph of the $(d{-}1)$- and $d$-simplices. Our algorithms differ in the specifics but use the same general strategy: use the tree decomposition to iteratively build the homologous or bounded chains. While the algorithms have seemingly different running times ($2^{O(k^d)}n$ for ours, $O(2^{2k}n)$ for theirs), the parameter $k$ is the treewidth of different graphs associated with a simplicial complex, so it is not obvious how the runtimes of algorithms compare. In Section \ref{sec:comparison}, we compare the treewidth of these graphs and find that in some cases the algorithms have comparable running times.
\par 
For the rest of this section, we will review the other previous work on the problems OBCP and OHCP.

\paragraph{OBCP.}
Dunfield and Hirani showed that OBCP is NP-hard \cite{dunfield-harani}. Borradaile, Maxwell, and Nayyeri showed OBCP over coefficients in $\Z_2$ is hard to approximate within some constant factor if $P\neq NP$ and hard to approximate with any constant factor assuming the Unique Games Conjecture, even for embedded complexes \cite{borradaile_mbc}. Blaser et al.~showed that OBCP is W[1]-hard when parameterized by solution size \cite{blaser_mbc}.
\par 
On the positive side, there are algorithms to solve OBCP in embedded complexes \cite{borradaile_mbc}, $d$-manifolds \cite{top_dim_obc}, with coefficients in $\R$ \cite{chambers_min_area}, in special cases with coefficients in $\Z$ \cite{dunfield-harani}, and with norms other than the Hamming norm \cite{cohensteiner_mbc}.
\par 
Our algorithm for OBCP can also be used to test whether two chains are homologous and whether a cycle is null-homologous. The naive algorithm for both of these problems is to solve a system of linear equations. Dey gives algorithms to test whether a cycle embedded on a surface is null-homologous that runs in $O(n+l)$ time, where $l$ is the length of the cycle, or with a $O(n)$ time preprocessing step, $O(g+l)$ time, where $g$ is the genus of the surface \cite{dey_null_homologous}. There is an algorithm for testing the homology class of a cycle that runs in $O(gl)$ time with a $O(n^\omega)$ preprocessing step, where $g$ is the rank of the homology group and $l$ is the number of simplices in the cycle ~\cite{busaryev_annotations}.

\paragraph{OHCP and Homology Localization.} Homology Localization is a special case of OHCP where the input chain is a cycle; thus, hardness results for homology localization hold for OHCP as well. Chen and Freedman showed that Homology Localization over coefficients in $\Z_{2}$ is NP-hard to approximate within any constant factor \cite{chen_freedman}. Chambers, Erickson, and Nayyeri showed that Homology Localization is NP-hard even for surface-embedded graphs \cite{chambers_min_cut}. Borradaile, Maxwell, and Nayyeri showed that OHCP is hard to approximate within any constant factor assuming the Unique Games Conjecture, even for embedded complexes \cite{borradaile_mbc}. Blaser and V\.agset showed Homology Localization is $W[1]$-hard when parameterized by solution size \cite{blaser_hl}.
\par
On the positive side, there are parameterized algorithms to solve homology localization for 1-cycles in surface embedded graphs \cite{chambers_min_cut, erickson_nayyeri_min_cut}, 1-cycles in simplicial complexes \cite{busaryev_annotations}, and to solve OHCP in $(d{+}1)$-manifolds \cite{borradaile_mbc}. There is an algorithm to solve OHCP with a norm other than the Hamming norm \cite{cohensteiner_mbc}. There are also linear programming based approaches to solve OHCP in special cases with coefficients in $\Z$ \cite{dey_optimal_homologous}.

\section{Background}
\label{sec:prelim}

In this section, we present the mathematical background needed for the rest of this paper.

\subsection{Simplicial complexes}

A \textit{\textbf{simplicial complex}} is a set $K$ such that (1) each element $\sigma\in K$ is a finite set and (2) for each $\sigma\in K$, if $\tau\subset\sigma$, then $\tau\in K$. An element $\sigma\in K$ is a \textit{\textbf{simplex}}. A simplex $\tau$ is a \textit{\textbf{face}} of a simplex $\sigma$ if $\tau\subset\sigma$. Likewise, $\sigma$ is a \textit{\textbf{coface}} of $\tau$. The simplices $\sigma$ and $\tau$ are \textit{\textbf{incident}}. 
\par
A simplex $\sigma$ with $|\sigma|=d+1$ is a \textit{\textbf{d-simplex}}. The set of all $d$-simplices in $K$ is denoted $K_d$. The \textit{\textbf{d-skeleton}} of $K$ is $K^{d}=\cup_{i=0}^{d}K_d$. In particular, the 1-skeleton of $K$ is a graph.  The \textit{\textbf{dimension}} of a simplicial complex is the largest integer $d$ such that $K$ contains a $d$-simplex. 
\par 
The union $V=\cup_{\sigma\in K}\sigma$ is the set of \textit{\textbf{vertices}} of a simplicial complex. Each simplex in $K$ is a subset of $V$. A subset of vertices $U\subset V$ defines a subcomplex of $K$. The \textit{\textbf{subcomplex induced by U}} is  $K[U]=\{\sigma\in K\suchthat\sigma\subset U\}$.

\subsection{Homology}

Let $K$ be a simplicial complex. The \textit{\textbf{d\textsuperscript{th} chain group}} $C_d(K)$ of $K$ is the free abelian group with coefficients over $\Z_2$ generated by $K_d$. An element $\gamma\in C_d(K)$ is a \textit{\textbf{d-chain}}. A $d$-chain $\gamma$ with coefficients over $\Z_2$ naturally corresponds to a set of $d$-simplices in $K$, where the set contains a $d$-simplex $\sigma$ if and only if the coefficient on $\sigma$ is 1. We overload notation and use $\gamma$ to also refer to both the chain and the set defined by the chain. The \textit{\textbf{weight}} of a chain $\gamma$ is the number of simplices the set $\gamma$ contains and is denote $\|\gamma\|$; alternatively, $\|\gamma\|$ is the \textit{\textbf{Hamming norm}} of $\gamma$.  Addition of two $d$-chains takes the symmetric difference of the sets corresponding to these chains.
\par
Let $\sigma$ be a $d$-simplex. The \textit{\textbf{boundary}} of $\sigma$ is the $(d-1)$-chain $\boundary\sigma=\sum_{v\in\sigma}\sigma\setminus\{v\}.$
The boundary map linearly extends the notion of boundary from $d$-simplices to $d$-chains. The \textit{\textbf{boundary map}} is the homomorphism $\boundary_d:C_d(K)\to C_{d{-}1}(K)$ such that $\boundary_d\gamma=\sum_{\sigma\in\gamma}\boundary\sigma$. When obvious, we drop the $d$ and simply write $\boundary_d$ as $\boundary$. A ($d{-}1$)-chain $\beta$ \textit{\textbf{bounds}} a $d$-chain $\gamma$ if $\boundary\gamma=\beta$, and the chain $\gamma$ \textit{\textbf{spans}} $\beta$. 
\par
The composition of boundary maps $\boundary_{d{-}1}\circ\boundary_d=0$, the zero map. As $C_d(K)$ is abelian, then $\im\boundary_{d}\subset\ker\boundary_{d{-}1}$ is a normal subgroup. The \textit{\textbf{(d-1)st homology group}} is the quotient group $H_{d{-}1}(K)=\ker\boundary_{d{-}1}/\im\boundary_{d}$. An element $\beta\in\ker\boundary_{d{-}1}$ is a \textit{\textbf{(d{-}1)-cycle}}. A $(d{-}1)$-cycle $\beta$ is \textit{\textbf{null-homologous}} if $\beta=\boundary\gamma$ for some $d$-chain $\gamma$. Two $(d{-}1)$-chains $b$ and $h$ are \textit{\textbf{homologous}} if $b+h$ is null-homologous.

\subsection{Tree decompositions}

Let $G=(V,E)$ be a graph. A \textit{\textbf{tree decomposition}} of $G$ is a tuple $(T,X)$, where $T=(I,F)$ is a tree with nodes $I$ and edges $F$, and $X=\{X_t\subset V\,|\,t\in I\}$ such that (1) $\cup_{t\in I}X_t=V$, (2) for any $\{v_1,v_2\}\in E$, $\{v_1,v_2\}\subset X_t$ for some $t\in I$, and (3) for any $v\in V$, the subtree of $T$ induced by the nodes $\{t\in I\mid v\in X_t\}$ is connected. A set $X_t$ is the \textit{\textbf{bag}} of T. The \textit{\textbf{width}} of $(T,X)$ is ${\max}_{t\in I}|X_t|+1$. The \textit{\textbf{treewidth}} of a graph $G$ is $\tw(G)$, the minimum width of any tree decomposition of $G$. Computing the treewidth of a graph is NP-hard \cite{arnborg_treewidth_hardness}, but there are algorithms to compute tree decompositions that are within a constant factor of the treewidth, e.g. see \cite{bodlaender2013ock}. 
\par
A \textit{\textbf{nice tree decomposition}} is a tree decomposition with a specified root $r\in I$ such that (1) $X_r=\emptyset$, (2) $X_l=\emptyset$ for all leaves $l\in I$, and  (3) all non-leaf nodes are either an introduce node, a forget node, or a join node, which are defined as follows. An \textit{\textbf{introduce node}} is a node $t\in I$ with exactly one child $t'$ such that for some $w\in V$, $w\notin X_{t'}$ and $X_t=X_{t'}\cup\{w\}$. We say $t$ \textit{introduces} $w$. A \textit{\textbf{forget node}} is a node $t\in I$ with exactly one child $t'$ such that for some $w\in V$, $w\notin X_t$ and $X_t\cup\{w\}=X_{t'}$. We say $t$ \textit{forgets} $w$. A \textit{\textbf{join node}} is a node $t\in I$ with two children $t',t''$ such that $X_{t}=X_{t'}=X_{t''}$. 
\begin{lemma}[Lemma 7.4 of \cite{Cygan_2015}]
\label{lem:kn-nodes}
Given a tree decomposition $(T,X)$ of width $k$ of a graph $G=(V,E)$, a nice tree decomposition of width $k$ with $O(kn)$ nodes can be computed in $O(k^2\max\{|V|,|I|\})$ time.
\end{lemma}
\begin{figure}
    \centering
    \begin{subfigure}{0.3\textwidth}
        \centering
        \includegraphics[height=1in]{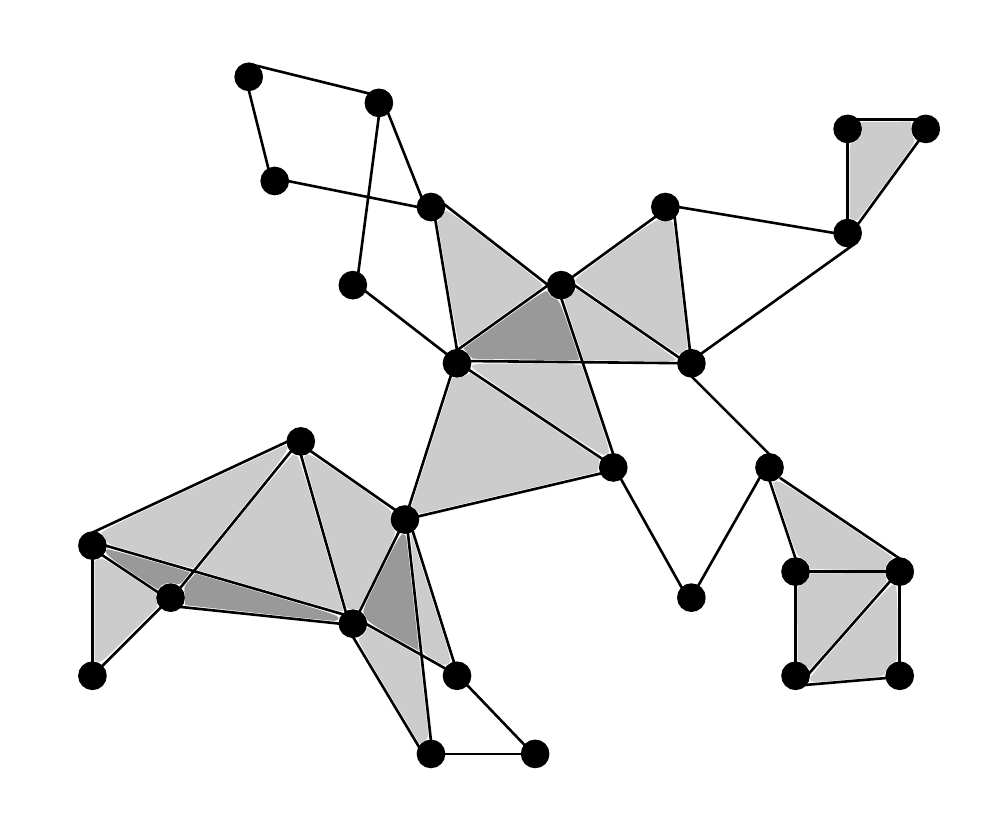}
    \end{subfigure}
    \begin{subfigure}{0.3\textwidth}
        \centering
        \includegraphics[height=1in]{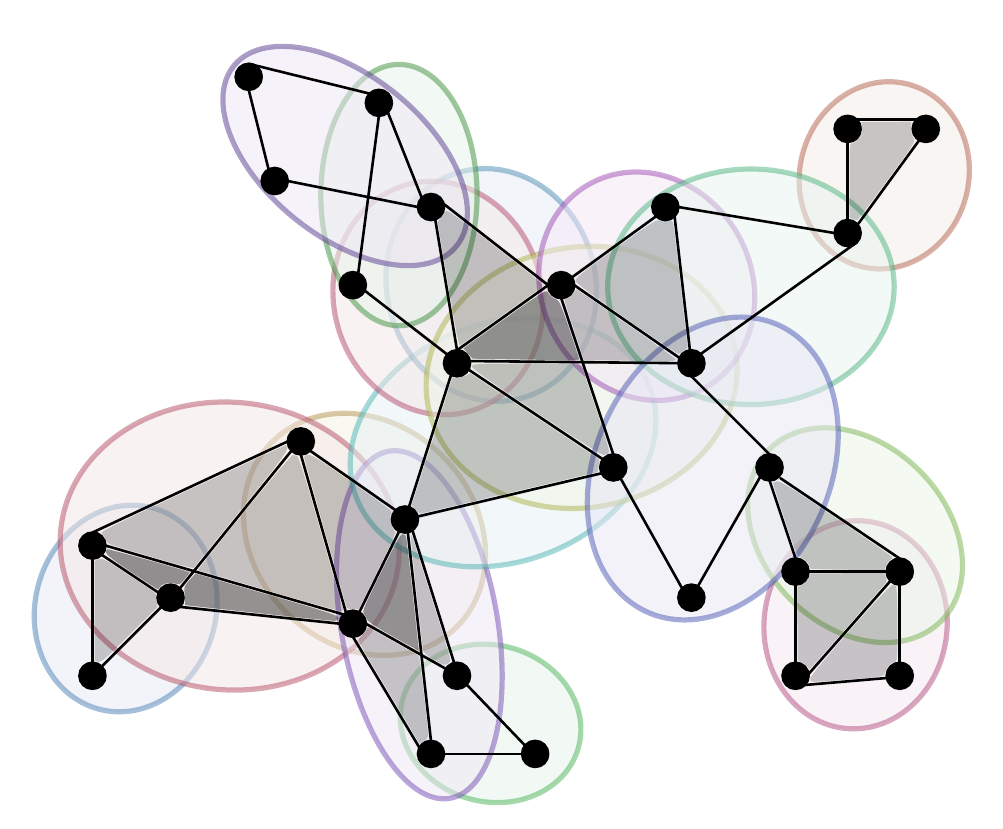}
    \end{subfigure}
    \begin{subfigure}{0.3\textwidth}
        \centering
        \includegraphics[height=1in]{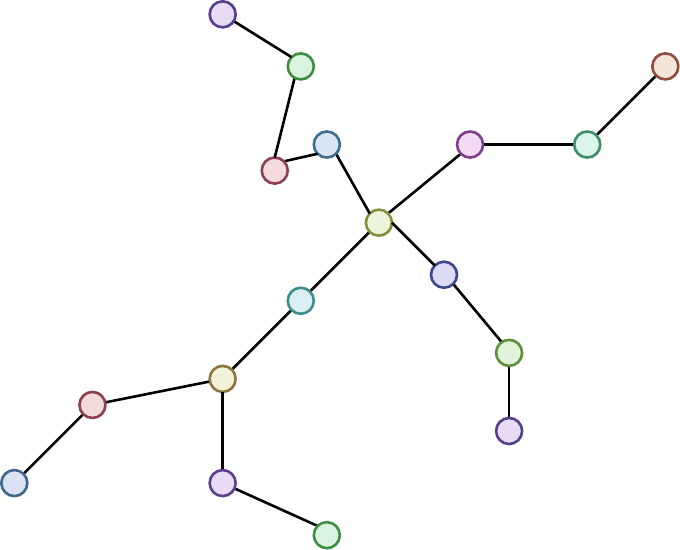}
    \end{subfigure}
    \caption{Left: A simplicial complex. Right and Center: A (not nice) tree decomposition of the simplicial complex. Each node of the tree corresponds to a subset of the vertices of the simplicial complex}
    \label{fig:td}
\end{figure}  
Let $K$ be a simplicial complex. Let $(T,X)$ be a tree decomposition of the 1-skeleton of $K$. Each vertex and edge of $K$ is contained in a bag of the tree decomposition by definition. In general, all simplices of $K$ are contained in a bag of the tree decomposition, which we prove in Corollary \ref{cor:simplex_bag}. 

\begin{lemma}
\label{lem:clique-bag}
Let $(T,X)$ be a tree decomposition of a graph $G=(V,E)$. If $Q\subset V$ forms a clique in $G$, then there is a node $t\in I$ such that $Q\subset X_t$. Moreover, the set of nodes $\{t\in I\suchthat Q\subset X_t\}$ form a connected subtree of $T$.
\end{lemma}
\begin{proof}
    A proof that $Q$ is contained in $X_t$ for some node $t$ can be found in \cite[Lemma 12.3.5]{diestel}. For the second claim, let $t_1$ and $t_2$ be nodes such $Q\subset X_{t_1}\cap X_{t_2}$. Let $t_3$ be a node on the unique path connecting $t_1$ and $t_2$. For any vertex $v\in Q$, $v\in X_{t_3}$ as $v\in X_{t_1}\cap X_{t_2}$ and the set of nodes whose bags contain $v$ form a connected subtree of $T$. So $Q\subset X_{t_3}$.
\end{proof}

\begin{corollary}
\label{cor:simplex_bag}
Let $(T,X)$ be a tree decomposition of the 1-skeleton $K^{1}$ of a simplicial complex $K$. Let $\sigma\in K$ be a simplex. There is a node $t\in I$ such that $\sigma\subset X_t$. Moreover, the set of nodes $\{t\in I\suchthat \sigma\subset X_t\}$ form a connected subtree of $T$.
\end{corollary}
\begin{proof}
All subsets of $\sigma$ are simplices in $K$. In particular, for any vertices $u,v\in\sigma$, the edge $\{u,v\}\in K$. Viewed as a set of vertices, $\sigma$ forms a clique in $K^{1}$.
\end{proof}
As all simplices of $K$ are contained in some bag in the tree decomposition, then the dimension of $K$ gives a lower bound on the treewidth of the 1-skeleton.
\begin{corollary}
\label{cor:tw_dim}
Let $K$ be a $d$-dimensional simplicial complex. Then $\tw(K^1)\geq d$.
\end{corollary}
In light of Corollary \ref{cor:simplex_bag}, we define a \textit{\textbf{tree decomposition of a simplicial complex K}} as a tree decomposition of the 1-skeleton of $K$. See Figure \ref{fig:td}. The terms nice tree decomposition, width of a tree decomposition, and treewidth of a simplicial complex are all defined analogously.

\section{The Generic Algorithm}
\label{sec:overview}

In this section, we present a generic version of the algorithm for our two problems. Later sections discuss the specific algorithms for each problem.
\par 
A tree decomposition $(T,X)$ of a simplicial complex $K$ defines a recursively nested series of subcomplexes of the complex. We can use this series of subcomplexes to recursively build solutions to our problems.
\par 
Specifically, subtrees of the tree $T$ define subcomplexes of $K$. For a node $t$ in the tree, let $V_t$ be the union of the bags of each of $t$'s descendants, including $t$ itself. If $K$ is a $d$-dimensional simplicial complex, we define the \textit{\textbf{subcomplex rooted at t}} to be $K_t:=K[V_t]\setminus (K[X_t])_d$. See Figure \ref{fig:Kt}. If $t'$ is a descendant of $t$, then $K_{t'}\subset K_{t}$. We can therefore recurse onto subcomplexes of $K$ by recursing onto subtrees $T$. 
\begin{figure}[H]
    \centering
    \begin{subfigure}{0.28\textwidth}
        \centering
        \includegraphics[height=1in]{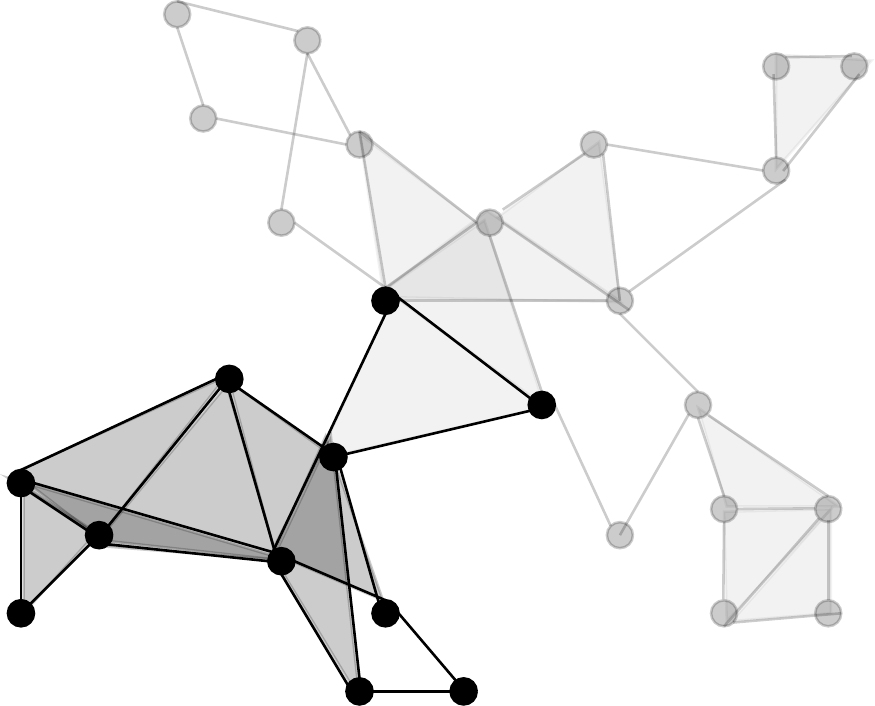}
    \end{subfigure}
    \begin{subfigure}{0.28\textwidth}
        \centering
        \includegraphics[height=1in]{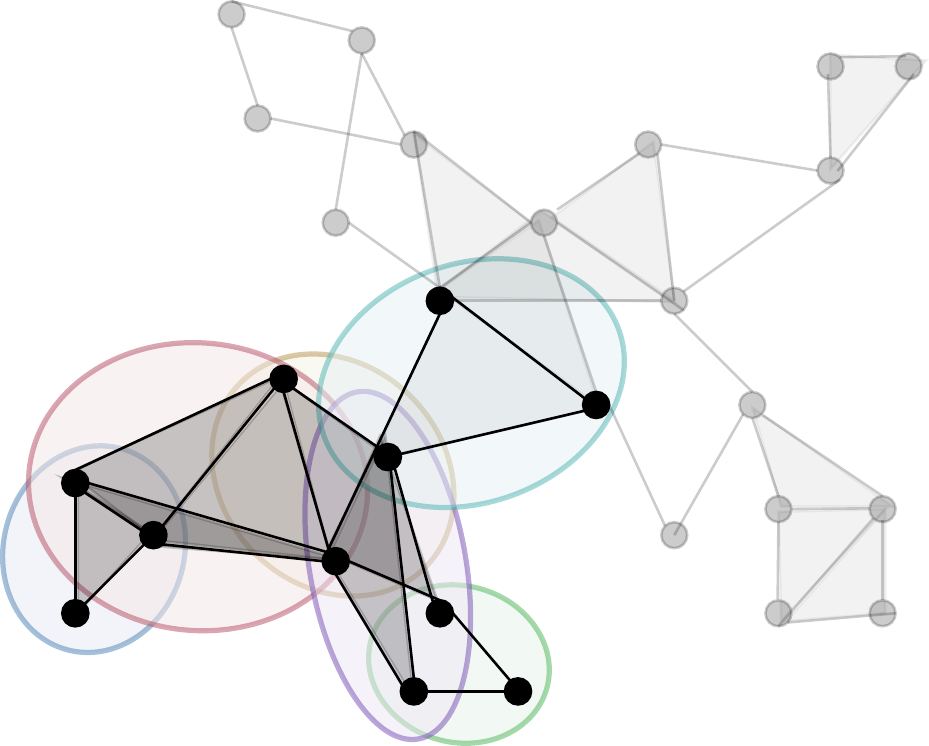}
    \end{subfigure}
    \begin{subfigure}{0.28\textwidth}
        \centering
        \includegraphics[height=1in]{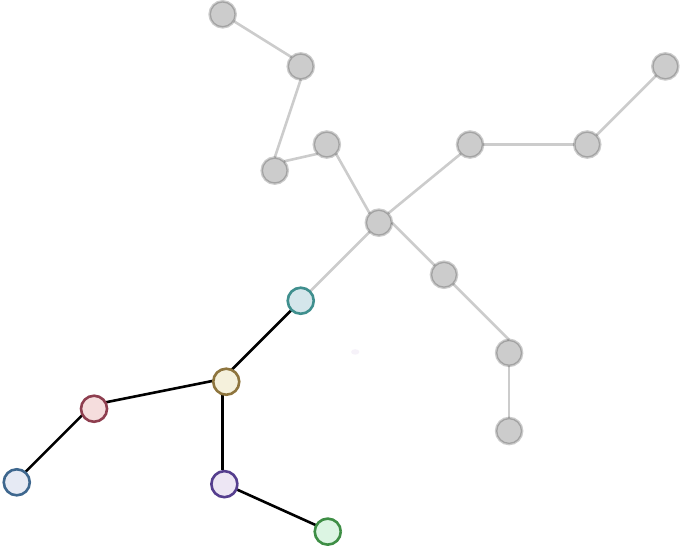}
    \end{subfigure}
    \caption{Right: A subtree rooted at $t$. Left: the corresponding subcomplex $K_t$.}
    \label{fig:Kt}
\end{figure}

\par
Each of our algorithms compute a set of \textit{candidate solutions} at each node $t$. The exact definition of candidate solution varies between problems, but intuitively, a candidate solution at a node $t$ is a $d$-chain in $K_t$ that could be a subset of an optimal solution. In each of our algorithms, we are able to define our candidate solutions recursively: if $\gamma$ is a candidate solution at $t$, then for each child $t'$ of $t$, $\gamma\cap K_{t'}$ is a candidate solution at $t'$. This is the key fact our algorithm uses to find candidate solutions at $t$. Our algorithm attempts to build candidate solutions at $t$ by adding $d$-simplices in $K_t\setminus K_{t'}$ to candidate solutions at $t'$.
\par
We compute the set of candidate solutions at the nodes in a bottom-up fashion, starting at the leaves and moving towards the root. Our general approach for computing the set of candidate solutions at $t$ is as follows. Take one candidate solution from each of $t$'s children. Take the union of these candidate solutions. Add a subset of $d$-simplices in $K_t$ that are not in $K_{t'}$ for any of $t$'s children $t'$ to this union. Check if this union is a candidate solution at $t$. Repeat for all sets of candidate solutions at the children and all subset of $d$-simplices.
\par
The above algorithm for computing the set of candidate solutions at a node is generally correct, but because we use a \textit{nice} tree decomposition, we won't perform every step of this algorithm at any one node. Instead, certain steps of the algorithm are performed only at specialized nodes. For example, we only take the union of candidate solutions from multiple children at join nodes, as join nodes are the only type of node with multiple children. Also, we only add a subset of $d$-simplices in $K_t\setminus K_{t'}$ to the candidate solutions at forget nodes. 
\par 
It might be counter-intuitive that forgetting a vertex adds $d$-simplices to $K_t$, so let's see why this is the case. Let $t$ be a forget node and $t'$ its unique child. The set of vertices $V_t=V_{t'}$, so $K[V_t]=K[V_{t'}]$. However, as the bags $X_{t'}\subsetneq X_t$, there may be $d$-simplices in $K[X_{t'}]$ that are not in $K[X_{t}]$. As we exclude the $d$-simplices in $K[X_{t'}]$ from $K_{t'}$, the $d$-simplices in $K[X_{t'}]\setminus K[X_t]$ will be in $K_t$ but not $K_{t'}$. See Figure \ref{fig:forget}. Alternatively, for introduce and join nodes, no new $d$-simplices are added to $K_t$ that do not appear in some $K_{t'}$ for a child $t'$ of $t$. If $t$ is an introduce node and $t'$ is its unique child, we can prove $(K_{t})_d=(K_{t'})_d$. While introducing a vertex $v$ may add $d$-simplices to $K[V_t]$, we can prove these simplices are always contained in $K[X_t]$; see Lemma \ref{lem:introduce_v}. Similarly, if $t$ is a join node and $t'$ and $t''$ are its two children, we can prove that $(K_t)_d=(K_{t'})_d\sqcup (K_{t''})_d$; see Lemma \ref{lem:mbc_join_chain}.

\begin{figure}
    \centering
    \begin{subfigure}{0.3\textwidth}
        \centering
        \includegraphics[width=0.6\linewidth]{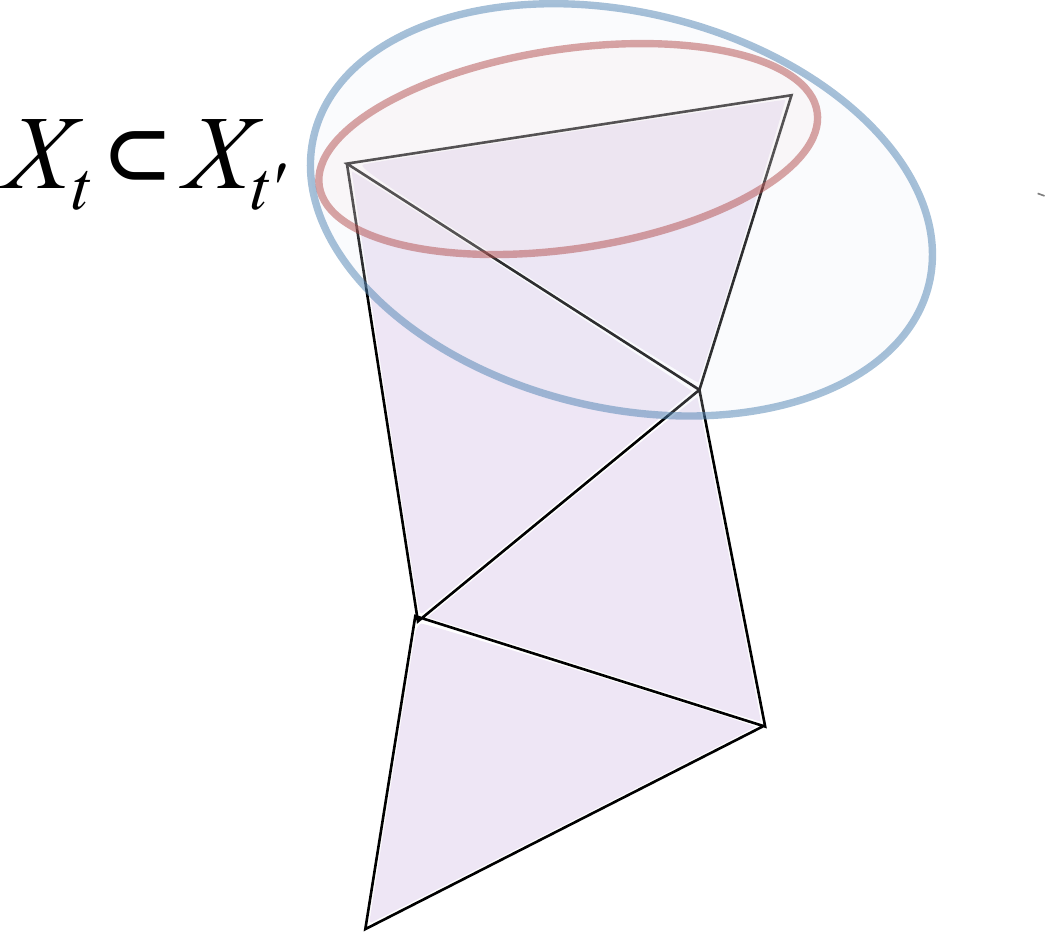}
    \end{subfigure}
    \begin{subfigure}{0.3\textwidth}
        \centering
        \includegraphics[width=0.6\linewidth]{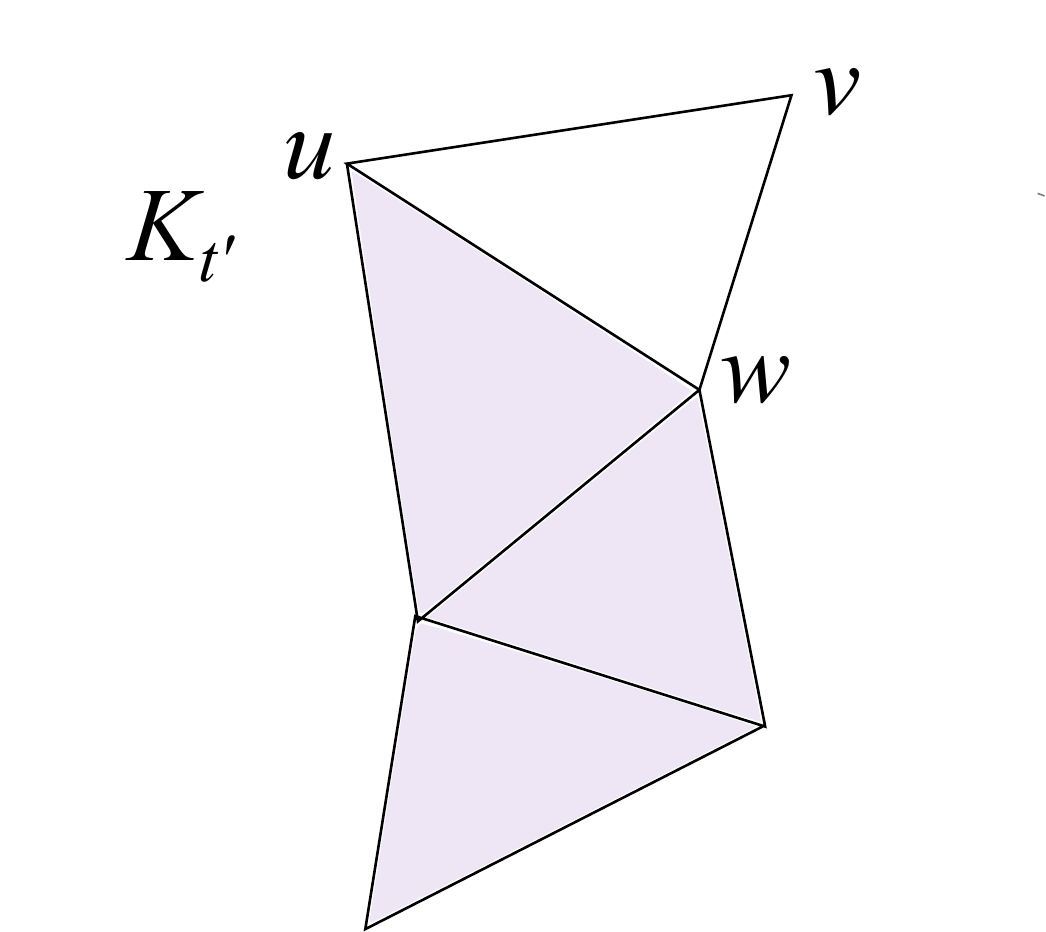}
    \end{subfigure}
    \begin{subfigure}{0.3\textwidth}
        \centering
        \includegraphics[width=0.6\linewidth]{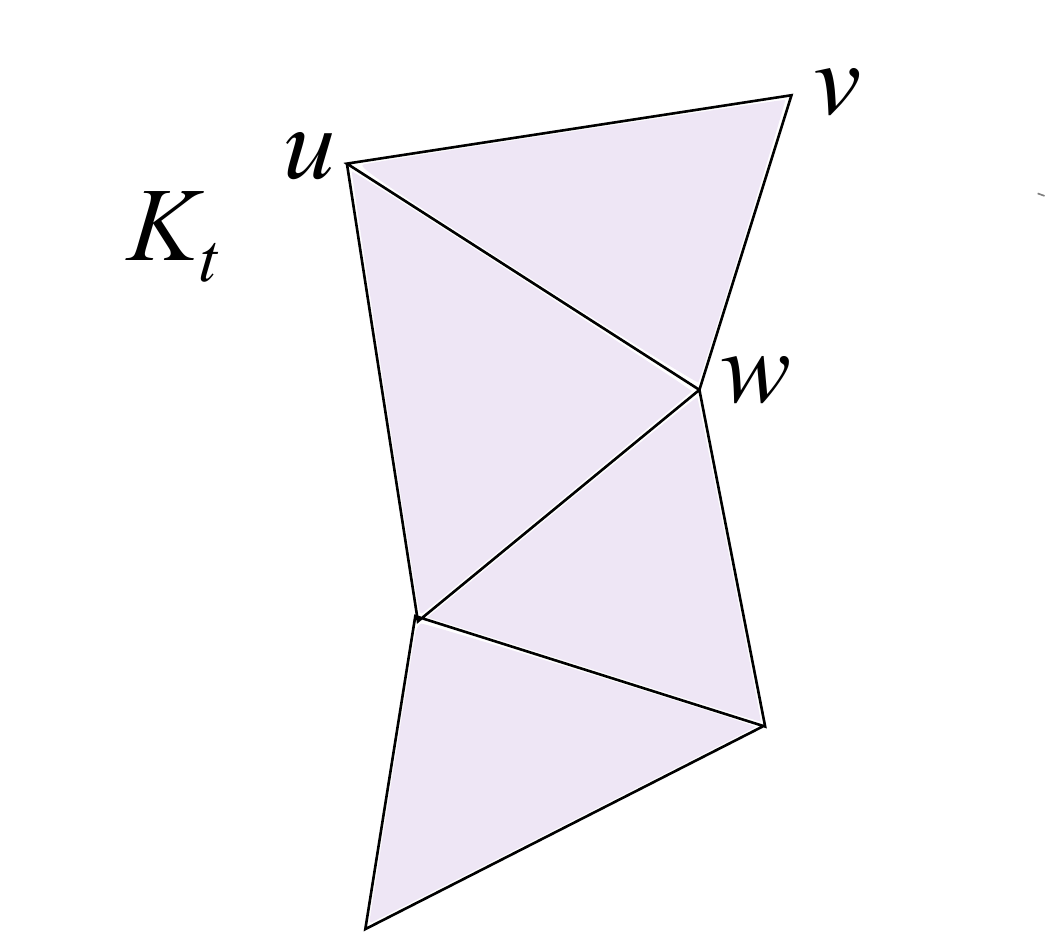}
    \end{subfigure}
    \caption{An example of a forget node $t$ and its child $t'$. The triangle $\{u,v,w\}$ is contained in $K[X_{t'}]$ but not $K[X_t]$. As $K_t$ excludes triangles in $K[X_t]$, then $\{u,v,w\}\in K_{t}\setminus K_{t'}$.}
    \label{fig:forget}
\end{figure}

%
%

\section{Optimal Bounded Chain Problem}
\label{section:mbc}

Let $K$ be a simplicial complex and $b\in C_{d{-}1}(K)$ a $(d{-}1)$-chain. The \textit{\textbf{Optimal Bounded Chain Problem}} asks to find the minimum weight $d$-chain $c\in C_{d}(K)$ bounded by $b$, if such a chain $c$ exists. Let $(T,X)$ a nice tree decomposition of $K$ with root $r$. We will use a dynamic program on $T$ to compute $c$.
\par
At a node $t\in T$, we are interested in the portion of $b$ contained in $K_t$; we define the \textit{\textbf{partial boundary at t}} as $b_t:=b\cap (K_t\setminus K[X_t])$. Our key observation is that any $d$-simplex that is incident to $b_t$ is contained in $K_t$.  This means that after the iteration in the algorithm where we process $t$, the portion of a candidate solution's boundary in $K_t$ but outside of $K[X_t]$ cannot be changed by adding more $d$-simplices. Therefore, a candidate solution $\gamma$ at $t$ should be a $d$-chain in $K_t$ that satisfies $\boundary\gamma\setminus K[X_t]=b_t$. Otherwise, we place no restriction on the rest of $\gamma$'s boundary, $\boundary\gamma\cap K[X_t]$. There might be other simplices in $K\setminus K_t$ incident to $\boundary\gamma\cap K[X_t]$, so even if $\boundary\gamma\cap K[X_t]\neq b\cap K[X_t]$, this portion of $\boundary\gamma$ might change later in the algorithm. We therefore let $\boundary\gamma\cap K[X_t]$ be any $(d{-}1)$-chain of $K[X_t]$. We call this chain $\beta=\boundary\gamma\cap K[X_t]$. A $d$-chain $\gamma\in C_d(K_t)$ is said to \textit{\textbf{span $\boldsymbol\beta$ at t}} if $\boundary\gamma\setminus K[X_t]=b_t$ and $\boundary\gamma\cap K[X_t]=\beta$. See Figure \ref{fig:obcp_example}.
 \par   
 \begin{figure}[H]
    \centering
    \begin{subfigure}{0.12\textwidth}
        \centering
        \includegraphics[height=0.7in]{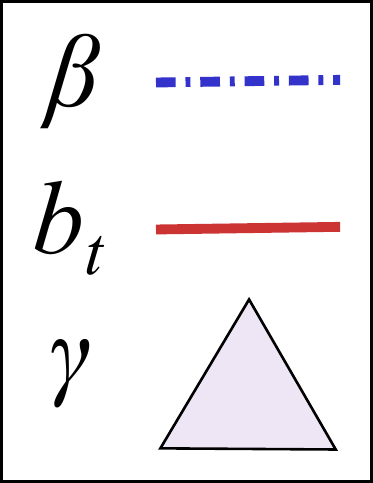}
    \end{subfigure}
    \begin{subfigure}{0.2\textwidth}
        \centering
        \includegraphics[height=1.3in]{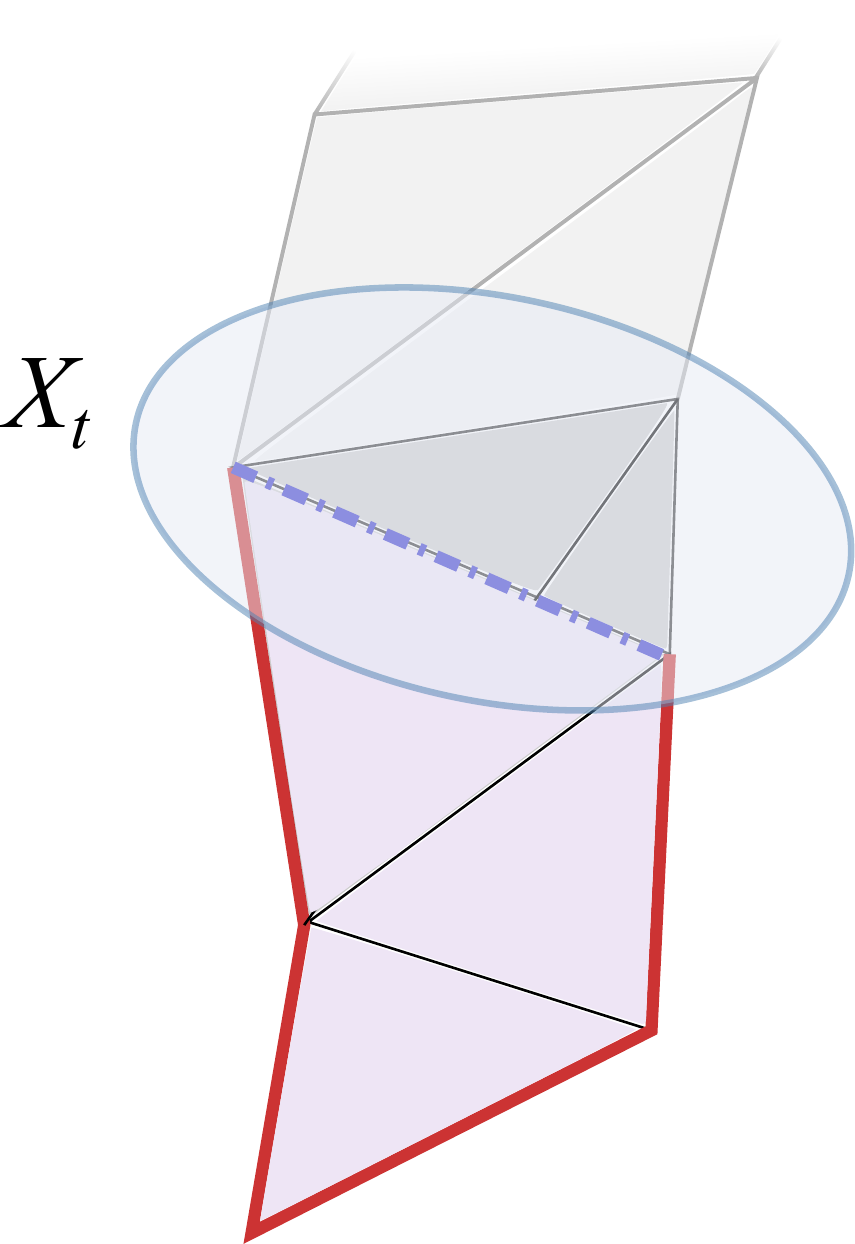}
    \end{subfigure}
    \begin{subfigure}{0.2\textwidth}
        \centering
        \includegraphics[height=1.3in]{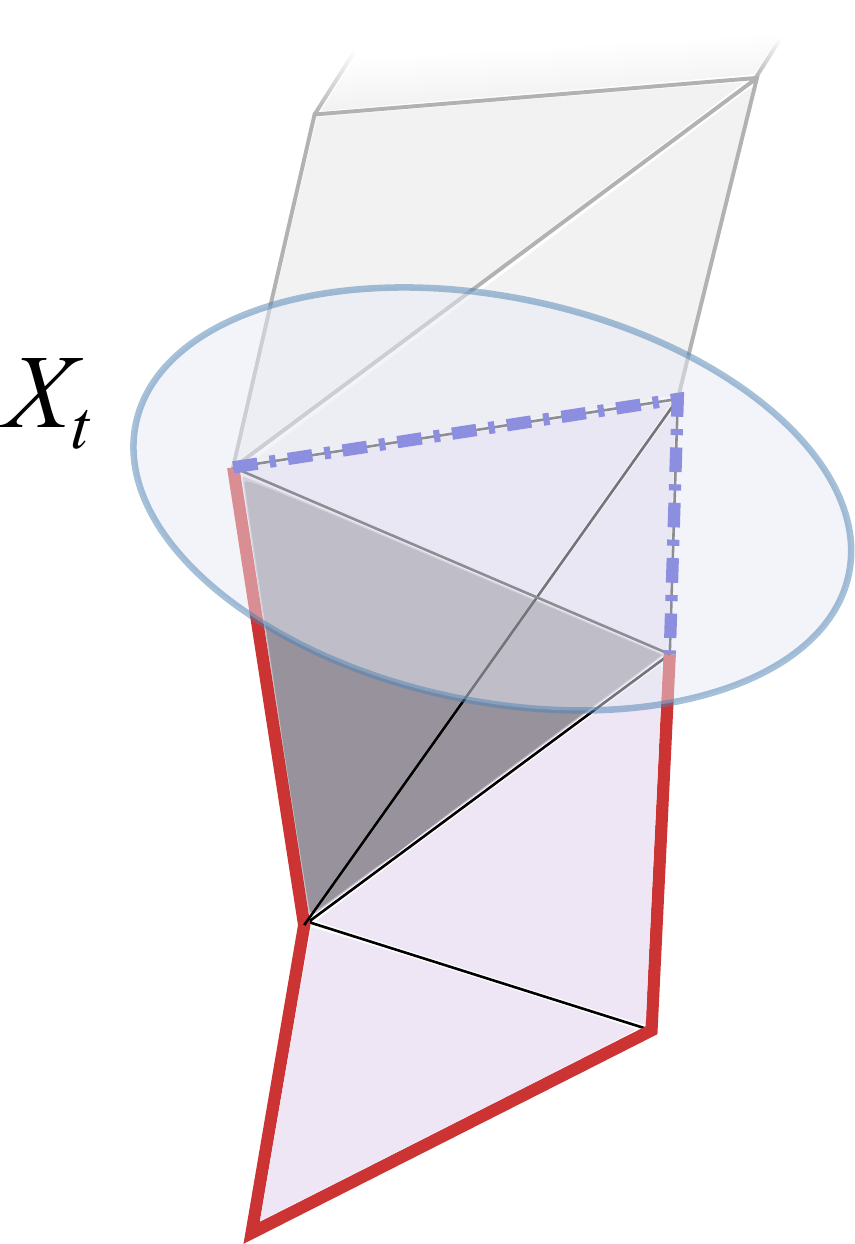}
    \end{subfigure}
    \caption{Two examples of candidate solutions for OBCP. Notice that $\boundary\gamma \cap K[X_t] = \beta$ is different for the two candidate solutions, but $\boundary\gamma\setminus K[X_t]=b_t$ is the same. In general, $\boundary\gamma \cap K[X_t]$ can be any $(d{-}1)$-chain in $K[X_t]$, but $\boundary\gamma\setminus K[X_t]$ must be $b_t$.}
    \label{fig:obcp_example}
\end{figure}
 
 For each $(d{-}1)$-chain $\beta$ in $K[X_t]$ such that $\beta \cup b_t$ is a boundary cycle, we store the minimum weight of a $d$-chain that spans $\beta$ at $t$ in the dynamic programming table entry $\V[\beta,t]$. The bounded width of the decomposition implies a bounded number of $(d{-}1)$-chains in $K[X_t]$, so the table $\V$ is of bounded size. Lemma \ref{lem:mbc_root} tells us that the table $\V$ contains the weight of the optimal solution.
\par
\begin{lemma}
\label{lem:mbc_root}
    The entry $\mathbf{V}[\emptyset,r]$ is the weight of the minimum weight chain $c$ such that $\boundary c=b$, if such a chain $c$ exists. 
\end{lemma}
\begin{proof}
     Recall that the bag $X_r$ of the root $r$ in our tree decomposition is empty, so $K[X_r]=\emptyset$. Also, each node in the tree decomposition is a descendant of $r$, so $V_r=V$, $K_r=K$, and $b_r=b$. There is a unique $(d{-}1)$-chain $\emptyset\in C_{d{-}1}(K[X_r])$. The entry $\mathbf{V}[\emptyset,r]$ is the minimum weight of a $d$-chain $c\in C_d(K_r)=C_d(K)$ such that (1) $\boundary c\cap X_r=\emptyset$ and, more importantly, (2) $\boundary c\setminus K[X_r]=\boundary c=b_r=b$. If no such chain $c$ exists, then $\mathbf{V}[\emptyset, r]=\infty$.
\end{proof}
\par\noindent
We are almost ready to present our algorithm for OBCP, but first, we need a lemma.
\begin{lemma}
\label{lem:simplex_in_subcomplex}
     If $\sigma\in K[V_t]$, then there is a descendant $t_{\sigma}$ of $t$ such that $\sigma\in K[X_{t_{\sigma}}]$. In particular, if $\sigma\in K[V_t]$, then there is a descendant $t_{\sigma}$ of $t$ such that $\sigma\in K[X_{t_{\sigma}}]$
\end{lemma}
\begin{proof}
    Recall that $K_t=K[V_t]\setminus (K[X_t])_d$, so $K_t\subset K[V_t]$. We prove the more general statement that if $\sigma\in K[V_t]$, then there is a descendant $t_{\sigma}$ of $t$ such that $\sigma\in K[X_{t_{\sigma}}]$. Let $\sigma\in K[V_t]$. By Lemma \ref{lem:clique-bag} there is a node $t_{\sigma}$ in $T$ such that $\sigma\subset X_{t_{\sigma}}$. If $t_{\sigma}$ is in the subtree rooted at $t$, we are done. So assume $t_{\sigma}$ is not in the subtree rooted at $t$. For every $v\in\sigma$, there is a node $t_v$ in the subtree rooted at $t$ such that $v\in X_{t_v}$. As well, $v\in X_{t_{\sigma}}$ as $v\in\sigma\subset X_{t_{\sigma}}$. The nodes containing $v$ form a connected subtree of $T$. So $v$ is contained in the bag of each node on the unique path connecting $t_v$ and $t_{\sigma}$. In particular, $v\in X_{t}$. Therefore, the bag $X_{t}$ contains every $v\in\sigma$, so $\sigma\in K[X_{t}]$.
\end{proof}
\par 
We now give our dynamic program to compute $\mathbf{V}$ on a nice tree decomposition. We compute the table $\V$ starting at the leaves of $T$ and moving towards the root. At each node $t$ in our tree decomposition, we calculate the entries $\mathbf{V}[\beta,t]$ using the entries of $\mathbf{V}$ of $t$'s children. Therefore, to specify our dynamic program, it suffices to specify how to calculate $\mathbf{V}$ at each type of node in a nice tree decomposition.

\subsection{Leaf Nodes}

Let $t$ be a leaf node. Recall that $X_t=\emptyset$, so $K[X_t]=\emptyset$. Moreover, $t$ has no children, so $V_t=K_t=b_t=\emptyset$. There is a unique $(d{-}1)$-chain $\beta=\emptyset\in C_{d{-}1}(K[X_t])$. There is also a unique $d$-chain $\gamma=\emptyset\in C_{d}(K_t)$. The chain $\gamma$ spans $\beta$ at $t$, so the unique table entry $\mathbf{V}[\emptyset,t]=0$.

\subsection{Introduce Nodes}

Let $t$ be an introduce node and $t'$ the unique child of $t$. Recall that $X_t=X_{t'}\sqcup\{w\}$ for some vertex $w\in K$. Our first observation is that any simplex in the subcomplex $K[V_t]$ that contains $w$ is only contained in the subcomplex at the bag $K[X_t]$.
\begin{lemma}
\label{lem:introduce_v}
Let $\sigma\in K[V_t]$ such that $w\in\sigma$. Then $\sigma\in K[X_t]\setminus K[V_{t'}]$.
\end{lemma}
\begin{proof}
 As $\sigma\in K[V_t]$, then by Lemma \ref{lem:simplex_in_subcomplex}, there is a node $t_\sigma$ in the subtree rooted at $t$ such that $\sigma\in K[X_{t_\sigma}]$. Suppose for the purposes of contradiction that $t\neq t_\sigma$. Then $t_\sigma$ is in the subtree rooted at $t'$. As $w\in\sigma$, then $w\in X_{t_\sigma}$. The set of nodes whose bags contain $w$ form a connected subtree of $T$, so each node on the unique path connecting $t$ and $t_\sigma$ must contain $w$ in its bag. The node $t'$ is on this path, but $w\notin X_{t'}$, a contradiction. Thus, $t=t_\sigma$. As no node in the subtree rooted at $t'$ contains $\sigma$, then by Lemma \ref{lem:simplex_in_subcomplex}, $\sigma\notin K[V_{t'}]$.
\end{proof}

As each simplex that contains $w$ is contained in $K[X_t]$, then introducing $w$ does not change the complexes $K_t$ outside of $K[X_t]$. We prove this in the following lemma.
\begin{lemma}
\label{lem:mbc_introduce_complex}
The complexes $K_{t'}\setminus K[X_{t'}]=K_{t}\setminus K[X_{t}]$.
\end{lemma}
\begin{proof}
    Observe that $K_t\setminus K[X_t]=(K[V_t]\setminus (K[X_t])_d)\setminus K[X_t]=K[V_t]\setminus K[X_t]$. We will therefore prove that $K[V_t]\setminus K[X_t]=K[V_{t'}]\setminus K[X_{t'}]$.
    \par
    Let $\sigma\in K[V_{t'}]\setminus K[X_{t'}]$. It follows that $\sigma\in K[V_{t}]$ as $V_{t'}\subset V_t$. We need to show that $\sigma\notin K[X_t]$. Suppose $\sigma\in K[X_t]$. As $X_t\setminus X_{t'}=\{w\}$, then $w\in\sigma$. This cannot be the case, as any simplex containing $w$ is not in $K_{t'}\subset  K[V_{t'}]$ by Lemma \ref{lem:introduce_v}. So $\sigma\notin K[X_t]$ and $\sigma\in K_t\setminus K[X_t]$
    \par
    Now let $\sigma\in K[V_{t}]\setminus K[X_{t}]$. We conclude that $\sigma\not\in K[X_{t'}]$ as $X_{t'}\subset X_t$. We need to show that $\sigma\in K[V_{t'}]$. Suppose $\sigma\not\in K[V_{t'}]$. Each descendant of $t$ is a descendant of $t'$ except $t$ itself. As $X_{t}\setminus X_{t'}=\{w\}$, then $V_{t}\setminus V_{t'}$ can only contain $w$. We conclude $w\in\sigma$. This is a contradiction, as any simplex containing $w$ in $K_t$ is contained in $K[X_t]$ by Lemma \ref{lem:introduce_v}. Thus $\sigma\in K[V_{t'}]\setminus K[X_{t'}]$.
\end{proof}

As introducing $w$ does not change $K_t$ outside of the bag $X_t$, then neither the chain group $C_d(K_t)$ nor the partial boundary $b_t$ are changed by introducing $w$. We prove this in the following two lemmas.

\begin{lemma}\label{lem:mbc_introduce_chain}
The chain groups $C_{d}(K_{t'})=C_d(K_{t})$.
\end{lemma}
\begin{proof}
    We prove this by showing that $(K_{t'})_d=(K_{t})_d$. This follows from Lemma \ref{lem:mbc_introduce_complex} as $(K_{t})_d=(K[V_{t}]\setminus (K[X_{t}])_d)_d=(K_t \setminus K[X_{t}])_d$.
\end{proof}
\begin{lemma}
\label{lem:mbc_introduce_boundary}
    The chain $b_t=b_{t'}$.
\end{lemma}
\begin{proof}
    This follows from Lemma \ref{lem:mbc_introduce_complex}, as 
    $
    b_t=b\cap K_t\setminus K[X_t] =b\cap K_{t'}\setminus K[X_{t'}]=b_{t'}.
    $
\end{proof}
\par\noindent
Introducing $w$ to the bag $X_t$ does not change $C_d(X_t)$ or $b_t$, so unsurprisingly, the values in the dynamic programming table don't change either. The following lemma proves this and gives a formula for computing $\V[\beta,t]$.
\begin{lemma}
\label{eqn:mbc-introduce} 
    Let $t$ be an introduce node and let $t'$ be the unique child of $t$. Let $\beta\in C_{d-1}(K[X_t])$. Then
    $$\mathbf{V}[\beta,t]=
    \begin{cases}
        \mathbf{V}[\beta,t']&\text{ if }\beta\in C_{d{-}1}(K[X_{t'}])\\
        \infty & \text{ otherwise.}
    \end{cases}$$
\end{lemma}
\begin{proof}
    Let $\beta\in C_{d{-}1}(K[X_{t'}])$. We claim that any chain $\gamma\in C_{d}(K_{t'})$ spanning $\beta$ at $t'$ spans $\beta$ at $t$ and vice versa. Let $\gamma$ be a $d$-chain that spans $\beta$ at $t'$. As $C_{d}(K_{t'})=C_{d}(K_{t})$, then $\gamma$ is a $d$-chain in both $K_t$ and $K_{t'}$; moreover, $\boundary\gamma\subset K_{t}\cap K_{t'}$, so $\boundary\gamma=\boundary\gamma\cap K_{t}=\boundary \gamma\cap K_{t'}$. Using this fact and Lemma \ref{lem:mbc_introduce_complex}, we see
    $$
        \boundary\gamma\setminus K[X_{t'}]=(\boundary\gamma\cap K_{t'})\setminus K[X_{t'}]=(\boundary\gamma\cap K_{t})\setminus K[X_t]=\boundary\gamma\setminus K[X_t].
    $$
    We use the fact that $\boundary\gamma\cap K[X_{t'}]=\boundary\gamma\cap K[X_t]$ to prove that $\boundary\gamma\cap K[X_t]=b_t$. Indeed, as $\boundary\gamma\setminus K[X_{t'}]=b_{t'}$ and $b_{t'}=b_{t}$ from Lemma \ref{lem:mbc_introduce_boundary}, then $\boundary\gamma\setminus K[X_{t}]=b_t$. Moreover, as $\boundary\gamma=\beta\sqcup b_{t}$, then the rest of the $\boundary\gamma$ is $\boundary\gamma\cap K[X_{t}]=\beta$. This completes the proof that $\gamma$ spans $\beta$ at $t$. By the same argument, any chain spanning $\beta$ at $t$ spans $\beta$ at $t'$. As the same set of chains span $\beta$ at $t$ and $t'$, then $\V[\beta,t]=\V[\beta,t']$
    \par
    Now let $\beta\in C_{d{-}1}(K[X_{t}])\setminus C_{d{-}1}(K[X_{t'}])$. As $\beta\notin C_{d{-}1}(K[X_{t'}])$, there must be a simplex $\sigma\in \beta$ such that $w\in\sigma$. As any simplex $\sigma$ such that $w\in\sigma$ is not contained in $K[X_{t'}]$, there is no chain $\gamma\in C_d(K_t)$ that spans $\beta$ at $t$ as $C_d(K_t)=C_d(K_{t'})$. Thus, $\V[\beta,t]=\infty$.
\end{proof}

\subsection{Forget Nodes}

Let $t$ be a forget node and $t'$ the unique child of $t$. Recall that $X_t\sqcup\{w\}=X_{t'}$. Forget nodes add new chains to the chain group $C_d(K_t)$. In particular, any $d$-simplex $\sigma\in (K[X_{t'}])_d$ that contains $w$ is not contained in $K_{t'}$ but will be contained in $K_t$. We prove this in our first lemma. Let $(K[X_{t'}])_d^{w}=\{\sigma\in (K[ X_{t'}])_d\suchthat w\in\sigma\}$ and let $C_d^{w}(K[X_{t'}])$ be the free abelian group on $(K[X_{t'}])_d^{w}$ with coefficients in $\Z_2$. 
\begin{lemma}
\label{lem:mbc_forget_chain}
    The chain group $C_d(K_t)=C_d(K_{t'})\oplus C_{d}^{w}(K[X_{t'}])$.
\end{lemma}
\begin{proof}
    As $C_d(K_t)$ is generated by $(K_t)_d$, we can prove the lemma by showing that $(K_t)_d=(K_{t'})_d\sqcup (K[X_{t'}])_{d}^{w}$.
    \par
    We first prove that $(K_{t'})_d\subset (K_t)_d$. Let $\sigma\in (K_{t'})_d$. By the definition of $K_{t'}$, $\sigma\in K[V_{t'}]$ and $\sigma\not\in K[X_{t'}]$. We know that $K[V_{t}]=K[V_{t'}]$ as each descendant of $t'$ is a descendant of $t$ and $X_t\subset X_{t'}$, so $\sigma\in K[V_{t}]$ as well. Furthermore, $\sigma\not\in K[X_{t}]$ as $K[X_{t}]\subset K[X_{t'}]$ and $\sigma\not\in K[X_{t'}]$. As $\sigma\in K[V_t]$ and $\sigma\not\in K[X_t]$, then $\sigma\in (K_{t})_d$ as $(K_t)_d=K[V_t]\setminus K[X_t]$. This proves that $(K_{t'})_d\subset (K_{t})_d$.
    \par
    As $(K_{t'})_d\subset (K_t)_d$, the next step to proving $(K_t)_d=(K_{t'})_d\sqcup (K[X_{t'}])_{d}^{w}$  is to show that $(K_t)_d\setminus (K_{t'})_d=(K[X_{t'}])_{d}^{w}$. Let $\sigma\in (K_t)_d\setminus (K_{t'})_d$. As $\sigma\in (K_t)_d$, then $\sigma\in K[V_t]$ and $\sigma\not\in K[X_t]$. As $\sigma\notin (K_t)_d$, then either $\sigma\not\in K[V_{t'}]$ or $\sigma\in K[X_{t'}]$. We know that $\sigma\in K[V_{t'}]$, as we proved in the previous paragraph that $K[V_{t'}]=K[V_t]$, so we conclude that $\sigma\in K[X_{t'}]$. It must be the case that $w\in\sigma$ as $\sigma\in K[X_{t'}] \subset K[X_t]$ and $X_{t'}\setminus X_t=\{w\}$, so $\sigma\in(K[X_{t'}])_{d}^{w}$. This proves that $(K_t)_d\setminus (K_{t'})_d \subset (K[X_{t'}])_{d}^{w}$
    \par 
    We now prove that $(K[X_{t'}])_{d}^{w}\subset (K_t)_d\setminus (K_{t'})_d$. Let $\sigma\in(K[X_{t'}])_{d}^{w}$. We know that $\sigma\in K[X_{t'}]$, which implies that $\sigma\in K[V_{t}]$ as $K[X_{t'}]\subset K[V_t]$. As $w\in\sigma$, we also have that $\sigma\not\in K[X_t]$, so $\sigma\in (K_t)_d$. Lastly, as $\sigma\in (K[X_{t'}])_d$, then $\sigma\notin (K_{t'})_d = K[V_{t'}]\setminus K[X_{t'}]$. Thus, $\sigma\in (K_{t})_d\setminus (K_{t'})_d$ and $(K[X_{t'}])_{d}^{w}\subset (K_t)_d\setminus (K_{t'})_d$.
\end{proof}
Intuitively, we can build the chain $\gamma$ that attains $\mathbf{V}[\beta, t]$ by composing a minimal chain $\gamma'\in C_d(K[X_{t'}])$ that attains $V[\beta', t']$ and a chain $\gamma_w\in C_d^w(K[X_{t'}])$ of $d$-simplices that all contain $w$. We need to enforce that $\boundary(\gamma'+\gamma_w)=\beta+b_t$. The chain $\gamma'$ includes $b_{t'}$ in the portion of its boundary outside $K[X_{t'}]$, so we want to choose $\gamma_w$ to help cover the rest of the boundary $b_t\setminus b_{t'}$. We can prove that $b_t\setminus b_{t'}$ is contained in $K[X_{t'}]$, i.e. $b_t\setminus b_{t'}=b_t\cap K[X_{t'}]$. So if $(\boundary(\gamma'+\gamma_w)\cap K[X_{t'}])\setminus K[X_t]=b_t\cap K[X_{t'}]$, then $\boundary(\gamma'+\gamma_w)\setminus K[X_t]$ will cover all of $b_t$. This is equivalently to requiring that $(\beta'+\boundary\gamma_w)\setminus K[X_{t}]=b_t\cap K[X_{t'}]$, as $\boundary\gamma'\cap K[X_t]=\beta'$ and $\boundary\gamma_w\cap K[X_t] = \gamma_w$. Indeed, this is the requirement we see in the formula for $\V[\beta,t]$ in Lemma \ref{lem:mbc_forget}.
\par 
We first give a characterization of $b_{t'}$.
\begin{lemma}
\label{lem:mbc_boundary_forget}
The chain $b_{t'}=b_{t}\setminus K[X_{t'}].$
\end{lemma}
\begin{proof}
The chains $b_t$ and $b_{t'}$ are defined $b_{t'}=(b\cap K_{t'})\setminus K[X_{t'}]$ and $b_{t}=(b\cap K_{t})\setminus K[X_t]$. If we consider the difference $b_{t}\setminus K[X_{t'}]$, we find that 
$$
b_{t}\setminus K[X_{t'}]=((b\cap K_{t})\setminus K[X_t])\setminus K[X_{t'}]=(b\cap K_{t})\setminus K[X_{t'}]
$$
as $K[X_{t}]\subset K[X_{t'}]$. We will use this fact later.
\par 
The sets of vertices in the trees rooted at $t$ and $t'$ are equal, i.e. $V_t=V_{t'}$, but the vertices in the bags $X_t\subset X_{t'}$. We use these two facts to show that the complexes $K_t$ and $K_{t'}$ only differ in the complexes induced by their bags, namely $K_t\setminus K[X_{t'}]=K_{t'}\setminus K[X_{t'}]$:
$$
K_t\setminus K[X_{t'}]=K[V_t]\setminus (K[X_t])_d\setminus K[X_{t'}]=K[V_t]\setminus K[X_{t'}]=K[V_{t'}]\setminus (K[X_{t'}])_d\setminus K[X_{t'}]=K_{t'}\setminus K[X_{t'}].
$$
We have proved that $b_t=b\cap (K_t\setminus K[X_{t'}])$. We have also proved that $K_t\setminus K[X_{t'}]=K_{t'}\setminus K[X_{t'}]$. We use these facts to prove that $b_t\setminus K[X_{t'}]=b_{t'}$ as
\begin{equation*}
b_{t}\setminus K[X_{t'}]=b\cap K_{t}\setminus K[X_{t'}]=b\cap K_{t'}\setminus K[X_{t'}]=b_{t'}.\qedhere    
\end{equation*}
\end{proof}
\par\noindent
Lemma \ref{lem:mbc_boundary_forget} implies that $b_{t}\setminus b_{t'}=b_t\cap K[X_{t'}]$, as claimed above. We are now ready to give a formula for $\V[\beta,t]$.
\begin{lemma}
\label{lem:mbc_forget}
Let $\beta\in C_{d{-}1}(K[X_t])$. Then
$$\mathbf{V}[\beta,t]=\underset{\beta',\gamma_w}{\min}\;\mathbf{V}[\beta',t']+\|\gamma_w\|$$
where the minimization ranges over $\beta'$ and $\gamma_w$ such that
\begin{itemize}
    \item $\beta'\in C_{d{-}1}(K[X_{t'}])$
    \item $\gamma_w\in C_{d}^{w}(K[X_{t'}])$
    \item $(\beta'+\boundary \gamma_w)\cap K[X_t]=\beta$
    \item $(\beta'+\boundary\gamma_w)\setminus K[X_t]=b_t\cap K[X_{t'}]$
\end{itemize}
\end{lemma}
\begin{proof}
First, we show that each $d$-chain on the right hand side of the equation in the lemma spans $\beta$ at $t.$
Let $\gamma'\in C_d(K_{t'})$ that attains $\mathbf{V}[\beta', t']$, and let $\gamma_w \in C_d^{w}(K[X_{t'}])$ such that $(\beta'+\boundary \gamma_w)\cap K[X_t]=\beta$ and $(\beta+\boundary\gamma_w)\setminus K[X_t]=b_t\cap K[X_{t'}]$
Let $\gamma =\gamma' + \gamma_w$; we show that $\gamma$ spans $\beta$ at $t$. The boundary
$$
\boundary\gamma=\boundary\gamma'+\boundary \gamma_w=b_{t'}+\beta'+\boundary\gamma_w.
$$
If $\gamma$ were to span $\beta$ at $t$, then $\boundary\gamma\cap K[X_t]=\beta$ and $\boundary\gamma\setminus K[X_t]=b_t$. The intersection $\boundary\gamma\cap K[X_t]$ is
\begin{align*}
    \boundary\gamma\cap K[X_t] =& (b_{t'}+\beta'+\boundary\gamma_w)\cap K[X_t] \\
    =& (\beta'+\boundary \gamma_w)\cap K[X_t] \tag{as $b_{t'}\cap K[X_{t'}]=\emptyset$}\\
    =& \beta \tag{by assumption} 
\end{align*}
The difference $\boundary\gamma\setminus K[X_{t}]$ is 
\begin{align*}
\boundary\gamma\setminus K[X_{t}] =& (b_{t'}+\beta'+\boundary\gamma_w)\setminus K[X_t] \\ =& b_{t'}+(\beta'+\boundary\gamma_w)\setminus K[X_t] \tag{as $b_{t'}\cap K[X_t]=\emptyset$}\\
=& b_{t'}+b_{t}\cap K[X_{t'}] \tag{by assumption} \\
=& b_{t}\setminus K[X_{t'}]+b_{t}\cap K[X_{t'}] \tag{by Lemma \ref{lem:mbc_boundary_forget}} \\
=& b_t
\end{align*}
This proves that $\gamma$ indeed spans $\beta$ at $t$. Moreover, as $\gamma'$ and $\gamma_w$ are disjoint, then $\|\gamma\|=\|\gamma'\|+\|\gamma_w\|$. As $\gamma'$ achieves $\mathbf{V}[\beta',t']$ by assumption, then $\|\gamma\|=\mathbf{V}[\beta',t']+\|\gamma_w\|$.
\par
Next, we verify that the chain that achieves $\mathbf{V}[\beta,t]$ is included in the right hand side of the equation. By Lemma \ref{lem:mbc_forget_chain}, $\gamma=\gamma'+\gamma_w$ for some $\gamma'\in C_d(K_{t'})$ and $\gamma_w\in C_d^{w}(K[X_{t'}])$. Let $\beta'=\boundary \gamma'\cap K[X_{t'}]$. We need to verify that $\gamma'$ spans $\beta'$ at $t'$ and achieves $\mathbf{V}[\beta',t']$, that $(\beta'+\boundary \gamma_w)\cap K[X_{t}]=\beta$, and that $(\beta'+\boundary\gamma_w)\setminus K[X_{t}]=b_{t}\cap K[X_{t'}]$.
\par 
We first prove that $\gamma'$ spans $\beta'$ at $t'$. We already know that $\boundary \gamma'\cap K[X_{t'}] = \beta'$, so we only need to prove that $\boundary\gamma'\setminus K[X_{t'}] = b_{t'}$. Indeed,
\begin{align*}
\boundary \gamma'\setminus K[X_{t'}] =& (\boundary\gamma'+\boundary\gamma_w)\setminus K[X_{t'}] \tag{as $\boundary\gamma_w\subset K[X_{t'}]$} \\
=&\boundary\gamma\setminus K[X_{t'}] \tag{as $\boundary\gamma = \boundary\gamma'+\boundary\gamma_w$}\\ 
=& (\beta+b_{t})\setminus K[X_{t'}] \tag{as $\gamma$ bounds $\beta$ at $t$} \\
=& b_{t}\setminus K[X_{t'}] \tag{as $\beta\subset K[X_{t}]$ and $K[X_t]\subset K[X_{t'}]$}\\
=& b_{t'} \tag{by Lemma \ref{lem:mbc_boundary_forget}}
\end{align*}
This proves that $\gamma'$ spans $\beta'$ at $t'$. We delay proving that $\gamma'$ achieves $\mathbf{V}[\beta',t']$ until the end of the proof.
\par
We now prove that $(\beta'+\boundary\gamma_w)\cap K[X_{t}]=\beta$. We see that
\begin{align*}
\beta =& \boundary\gamma\cap K[X_t] \tag{by assumption}\\
=& (\boundary\gamma'+\boundary\gamma_w)\cap K[X_t] \tag{as $\boundary\gamma = \boundary\gamma'+\boundary\gamma_w$}\\
=& (b_{t'}+\beta'+\boundary\gamma_w)\cap K[X_t] \tag{as $\gamma'$ bounds $\beta'$ at $t'$} \\
=& (\beta'+\boundary\gamma_w)\cap K[X_t] \tag{as $b_{t'}\cap K[X_t]=\emptyset$.}
\end{align*}
Finally, we prove that $(\beta'+\boundary\gamma_w)\setminus K[X_{t}]=b_{t}\cap K[X_{t'}]$. 
\begin{align*}
 b_t\cap K[X_{t'}] =& (\boundary\gamma\setminus K[X_t])\cap K[X_{t'}] \tag{as $\gamma$ spans $\beta$ at $t$} \\
 = & (\boundary\gamma'+\boundary\gamma_w)\setminus K[X_t]\cap K[X_{t'}] \tag{as $\boundary\gamma = \boundary\gamma' + \boundary\gamma_w$}\\ 
 = & (b_{t'}+\beta'+\boundary\gamma_w)\setminus K[X_t]\cap K[X_{t'}] \tag{as $\gamma'$ bounds $\beta'$ at $t'$} \\ 
 = & (\beta'+\boundary\gamma_w)\setminus K[X_t] \tag{as $b_{t'}\cap K[X_{t'}]=\emptyset$}
\end{align*}
\par 
We are now ready to prove that $\gamma'$ achieves $\mathbf{V}[\beta',t']$. Suppose not, and let $\gamma'_{o}\in C_d(K_{t'})$ be a chain that achieves $V[\beta',t']$. We showed in the previous paragraph that $(\beta'+\boundary\gamma_w)\cap K[X_{t}]=\beta$ and $(\beta'+\boundary\gamma_w)\setminus K[X_{t}]=b_{t}\cap K[X_{t'}]$. This implies that $\gamma'_{o}+\gamma_w$ spans $\beta$ at $t$ by the first half of this proof. Moreover, $\|\gamma'_{o}+\gamma_w\|<\|\gamma'+\gamma_w\|=\|\gamma\|$, which contradicts the assumed optimality of $\gamma$. Thus, $\gamma'$ achieves $V[\beta',t']$ and $\gamma$ is included in the right hand side of the equation.
\end{proof}

\subsection{Join Nodes}

Let $t$ be a join node, and let $t'$ and $t''$ be the two children of $t$. Recall that $X_t=X_{t'}=X_{t''}$. We first observe that $K_{t'}$ and $K_{t''}$ only overlap in $K[X_t]$. 
\begin{lemma}
\label{lem:mbc_join_complex}
The complex $K_t\setminus K[X_t]=(K_{t'}\setminus K[X_t'])\sqcup(K_{t''}\setminus K[X_{t''}]).$
\end{lemma}
\begin{proof}
    As $K[X_t]=K[X_{t'}]=K[X_{t''}]$, we see immediately that both 
    $K_{t'}\setminus K[X_{t'}]\subset K_{t}\setminus K[X_t]$ and $K_{t''}\setminus K[X_{t''}]\subset K_{t}\setminus K[X_{t}]$ as both $K_{t'}\subset K_{t}$ and $K_{t''}\subset K_{t}$. 
    \par 
    We now prove that $K_{t'}\setminus K[X_{t'}]$ and $K_{t''}\setminus K[X_{t''}]$ are disjoint. Suppose there is a simplex $\sigma\in (K_{t'}\setminus K[X_{t'}])\cap (K_{t''}\setminus K[X_{t''}])$. Then by Lemma \ref{lem:simplex_in_subcomplex}, there are nodes $t_{\sigma}'$ and $t_{\sigma}''$ in the subtrees rooted at $t'$ and $t''$ such that $\sigma\in K[X_{t_{\sigma}'}]$ and $\sigma\in K[X_{t_{\sigma}''}]$. The set of nodes containing $\sigma$ form a connected subtree. This is a contradiction, as $t$ lies on the unique path connecting $t_{\sigma}'$ and $t_{\sigma}''$ and $\sigma\notin K[X_t]$ by assumption. Hence $K_{t'}\setminus K[X_{t''}]\cap K_{t''}\setminus K[X_{t''}]=\emptyset$.
    \par
    We now prove that $K_t\setminus K[X_t]\subset (K_{t'}\setminus K[X_{t'}])\sqcup(K_{t''}\setminus K[X_{t''}])$. Let $\sigma\in K_t\setminus K[X_t]$. By Lemma \ref{lem:simplex_in_subcomplex}, there is a node $t_\sigma$ in the subtree rooted at $t$ such that $\sigma\in K[X_{t_\sigma}]$. We know $t_\sigma\neq t$ as $\sigma\notin K[X_t]$, so $t_\sigma$ is either in the subtree rooted at $t'$ or the subtree rooted at $t''$.
\end{proof}

Lemma \ref{lem:mbc_join_complex} implies that any $d$-chain in $K_t$ is the sum of a $d$-chain in $K_{t'}$ and a $d$-chain in $K_{t''}$. We prove this in the following lemma.

\begin{lemma}\label{lem:mbc_join_chain}
    The chain group $C_d(K_{t})=C_d(K_{t'})\oplus C_d(K_{t''})$.
\end{lemma}
\begin{proof}
    As $C_d(K_{t})$ is generated by $(K_{t})_d$, we can prove this by showing that $(K_{t})_d=(K_{t'})_d\sqcup(K_{t''})_d$. We first note that 
    $$
        (K_{t})_d=(K[V_t]\setminus (K[X_t])_d)_d=(K[V_t]\setminus K[X_t])_d
    $$
    so it follows by Lemma \ref{lem:mbc_join_complex} that 
    \begin{equation*}
        (K_{t})_d=(K[V_t]\setminus K[X_t])_d=(K[V_{t'}]\setminus K[X_{t'}])_d\sqcup (K[V_{t''}]\setminus K[X_{t''}])_d=(K_{t'})_d\sqcup(K_{t''})_d.\qedhere
    \end{equation*}
\end{proof}
\par\noindent
Similarly, the boundary $b_t$ is composed of a portion in $K_{t'}$, namely $b_{t'}$, and a portion in $K_{t''}$, namely $b_{t''}$.

\begin{corollary}
\label{cor:mbc_join_boundary}
The chain $b_t=b_{t'}\sqcup b_{t''}.$
\end{corollary}
\begin{proof}
     The boundary $b_t=b\cap K_t\setminus K[X_t]$. Using Lemma \ref{lem:mbc_join_complex},
     \begin{equation*}
    b_t=b\cap K_t\setminus K[X_{t}]=(b\cap K_{t'}\setminus K[X_t])\sqcup(b\cap K_{t''}\setminus K[X_t])=b_{t'}\sqcup b_{t''}.\qedhere     
     \end{equation*}
\end{proof}
These two lemmas tell us that a chain $\gamma$ that spans $\beta$ at $t$ is the sum of chains $\gamma'$ and $\gamma''$. These chains satisfy three conditions: $\boundary\gamma'\setminus K[X_{t'}]=b_{t'}$, $\boundary\gamma''\setminus K[X_{t''}]=b_{t''}$, and  $(\boundary\gamma'+\boundary\gamma'')\cap K[X_t]=\beta$. To find the \textit{minimal} chain that spans $\beta$ at $t$, we only need to consider minimal chains $\gamma'$ and $\gamma''$ with properties. The following formula for $\mathbf{V}[\beta,t]$ confirms this.
\begin{lemma}
\label{eqn:mbc-join}
Let $\beta\in C_{d{-}1}(K[X_t])$. Then 
$$
\mathbf{V}[\beta,t]=\underset{\beta',\beta''}{\min} \;
\mathbf{V}[\beta',t']+\mathbf{V}[\beta'',t'']
$$
where the minimization ranges over $\beta'$ and $\beta''$ such that
\begin{itemize}
    \item $\beta'\in C_{d{-}1}(K[X_{t'}])$
    \item $\beta''\in C_{d{-}1}(K[X_{t''}])$
    \item $\beta'+ \beta''=\beta$
\end{itemize}
\end{lemma}
\begin{proof}
    We first show that each chain on the right hand side of the equation spans $\beta$ at $t$. Let $\gamma'\in C_d(K_{t'})$ and $\gamma''\in C_d(K_{t''})$ such that $\gamma'$ spans $\beta'$ at $t'$, $\gamma''$ spans $\beta''$ at $t''$, and $\beta'+ \beta''=\beta$. We claim that $\gamma=\gamma'+ \gamma''$ spans $\beta$ at $t$. This follows as 
    \begin{align*}
        \boundary\gamma\cap K[X_{t}] =& (\boundary\gamma'\cap K[X_t])+(\boundary \gamma''\cap K[X_{t}]) \\ =& \beta'+\beta'' \\
        =& \beta.
    \end{align*}
    and
    \begin{align*}
        \boundary\gamma\setminus K[X_{t}] =& (\boundary\gamma'\setminus K[X_{t}])+ (\boundary\gamma''\setminus K[X_{t}]) \\ 
        =& b_{t'}+b_{t''} \\
        =& b_t \tag{by Corollary \ref{cor:mbc_join_boundary}}
    \end{align*}
     Moreover, $\|\gamma\|=\|\gamma'\|+\|\gamma''\|$  as $\gamma'$ and $\gamma''$ are disjoint. If $\gamma'$ and $\gamma''$ achieve $\mathbf{V}[\beta',t']$ and $\mathbf{V}[\beta'',t'']$, then $\|\gamma\|=\mathbf{V}[\beta',t]+\mathbf{V}[\beta'',t'']$.
    \par
    Now let $\gamma\in C_d(K_t)$ be the minimum weight chain that spans $\beta$ at $t$. By Lemma \ref{lem:mbc_join_chain}, $\gamma=\gamma'+\gamma''$ for some $\gamma'\in C_d(K_{t'})$ and $\gamma''\in C_d(K_{t''})$. Let $\beta'=\boundary\gamma'\cap K[X_{t'}]$ and $\beta''=\boundary\gamma''\cap K[X_{t''}]$. We will show that $\gamma'$ and $\gamma''$ achieve $\mathbf{V}[\beta',t']$ and $\mathbf{V}[\beta'',t'']$ respectively. 
    \par 
    We first show that $\gamma'$ spans $\beta'$ at $t'$. We already know that $\boundary\gamma'\cap K[X_{t'}]=\beta'$, so we only need to prove that $\boundary\gamma'\setminus K[X_{t'}]=b_{t'}$. We consider the intersection $(\boundary\gamma\setminus K[X_t])\cap K_{t'}$. We see that 
    \begin{align*}
    (\boundary\gamma\setminus K[X_t])\cap K_{t'} =& (\boundary\gamma'\setminus K[X_t])\cap K_{t'}+(\boundary\gamma''\setminus K[X_t])\cap K_{t'} \\
    =& (\boundary\gamma'\setminus K[X_t])\cap K_{t'} \tag{*} \\
    =& \boundary\gamma'\setminus K[X_t] \tag{as $\boundary\gamma'\subset K_{t'}$.}
    \end{align*}
    where line (*) follows from the fact that $(\boundary\gamma''\setminus K[X_t])\cap K_{t'}=\emptyset$. This is the case as $\boundary\gamma''\setminus K[X_t]\subset K_{t''}\setminus K[X_t]$ and $K_t'\setminus K[X_{t'}]$ and $K_{t''}\setminus K[X_{t''}]$ are disjoint by Lemma \ref{lem:mbc_join_complex}.
    \par 
     We can alternatively express $(\boundary\gamma\setminus K[X_t])\cap K_{t'}$ as
    \begin{align*}
        (\boundary\gamma\setminus K[X_t])\cap K_{t'} =& b_t\cap K_{t'} \\
        =& b_{t'}\cap K_{t'}+b_{t''}\cap K_{t''} \tag{by Corollary \ref{cor:mbc_join_boundary}} \\
        =& b_{t'}\cap K_{t'} \tag{*}\\
        =& b_{t'} \tag{as $b_{t'}\subset K_{t'}$.} 
    \end{align*}
     where (*) follows as $b_{t''}\cap K_{t'}=\emptyset$. This is true as $b_{t''}\subset K_{t''}\setminus K[X_t]$ and $K_t'\setminus K[X_{t'}]$ and $K_{t''}\setminus K[X_{t''}]$ are disjoint by Lemma \ref{lem:mbc_join_complex}. Thus $\boundary\gamma'\setminus K[X_t]=(\boundary\gamma\setminus K[X_t])\cap K_{t'}=b_{t'}$ and $\gamma'$ spans $\beta'$ at $t'$. The proof that $\gamma''$ spans $\beta''$ at $t''$ is symmetrical. 
    \par
    As $\gamma$ is the minimum weight chain that spans $\beta$ at $t$, then $\gamma'$ and $\gamma''$ must be the chains that achieve $\mathbf{V}[\beta',t']$ and $\mathbf{V}[\beta'',t'']$ respectively. We saw in the first paragraph that any two chains $\gamma'$ and $\gamma''$ that span $\beta'$ at $t'$ and $\beta''$ at $t''$ sum to span $\beta$ at $t$. So if (say) $\gamma'$ did not achieve $\mathbf{V}[\beta',t']$, we could replace $\gamma'$ with the chain $\gamma'_{o}$ that achieves $\mathbf{V}[\beta',t']$ and $\gamma'_{o}+\gamma''$ would be a chain with strictly lower weight that spans $\beta$ at $t$, contradicting the assumed optimality of $\gamma$.
\end{proof}

\subsection{Analysis}

We now analyze the running time of the dynamic program, thereby proving Theorem~\ref{thm:obc}.

\begin{proof}[Proof of Theorem~\ref{thm:obc}]
We first show that we can compute the dynamic programming table entry for each type of node in $T$ in $2^{O\left(k^{d}\right)}$ time.
\par
 We can compute the entry at the leaf node in constant time as we are entering in a single table entry.
 \par 
 We can compute the value $V[\beta,t]$ of each $(d{-}1)$-chain $\beta$ at an introduce node in constant time. As there are $O(\binom{k+1}{d})=O\left((k+1)^{d}\right)$ $(d{-}1)$-simplices in $K[X_t]$, then there are $2^{O\left(k^{d}\right)}$ $(d{-}1)$-chains in $C_{d{-}1}(K[X_t])$. Processing an introduce node thus takes $2^{O\left(k^{d}\right)}$ time in total. 
 \par 
 We compute the entry at a forget node by performing nested iterations over $C_{d{-}1}(K[X_t])$ and $C_d^w(K[X_t])$. There are $O(\binom{k+1}{d})=O\left((k+1)^{d}\right)$ $(d{-}1)$-simplices in $K[X_t]$, so there are $2^{O\left(k^{d}\right)}$ $(d{-}1)$-chains $\beta\in C_{d{-}1}(K[X_t])$.  There are likewise $O(\binom{k+1}{d})=O(k^{d})$ $d$-simplices (we choose $d$ vertices in addition to $w$) in $(K[X_{t'}])_{d}^{w}$ and $2^{O(k^{d})}$ chains $\gamma_{w}\in C_{d}^{w}(K[X_{t'}])$. The size of both of the chain groups $C_{d{-}1}(K[X_{t'}])$ and $C_{d}^{w}(K[X_{t'}])$ is at most $2^{O\left(k^{d}\right)}$, so the nested iteration over both of these groups takes $2^{O\left(k^{d}\right)}$ time. 
 \par
 We compute the entry at a join node by iterating over the chain groups $C_{d{-}1}(K[X_{t'}])$ and $C_{d{-}1}(K[X_{t''}])$, which takes $2^{O\left(k^{d}\right)}$ time. 
 \par
 There are $O(kn)$ nodes in a nice tree decomposition, and the dynamic programming table entry at each node can be computed in $2^{O\left(k^{d}\right)}$ time. The algorithm therefore takes $2^{O\left(k^{d}\right)}n$ time in total. As $d\leq k$ by Corollary \ref{cor:tw_dim}, the running time of our algorithm is also $2^{O\left(k^{k}\right)}n$.
\end{proof}

As a corollary, we can test whether two $d{-}1$ chains $b$ and $h$ are homologous by running our algorithm on $b+h$. The chains $b$ and $h$ are homologous if and only if $\mathbf{V}[\emptyset,r]<\infty$. Hence, we obtain Corollary~\ref{cor:homology_test}. This implies we can use our algorithm to test whether or not a ($d{-}1$)-cycle $b$ is null-homologous, which is to say, whether $b$ is homologous to the empty chain.
\par
 Our algorithm uses the Hamming norm to measure the weight of a chain, but our algorithm can be easily adapted to solve OBCP in a simplicial complex with weight $d$-simplices. If $w:K_d\to\R^{+}$ is a weight function on the $d$-simplices, we define a weight function on $d$-chains $w:C_d(K)\to\R^{+}$ such that $w(\gamma)=\sum_{\sigma\in\gamma}w(\sigma)$ is the weight of the chain $\gamma$. If we substitute $w(\gamma)$ for $\|\gamma\|$ everywhere in our algorithm for OBCP, it is easy to see this adapted algorithm finds the minimum weight chain $c$ bounded by $b$.

%
%

\section{Optimal Homologous Chain Problem}

Let $K$ be a simplicial complex and let $b\in C_{d{-}1}(K)$. The \textbf{\textit{Optimal Homologous Chain Problem}} asks to find the minimum weight $(d{-}1)$-chain $h\in C_{d{-}1}(K)$ such that $b+h=\boundary c$ for some $c\in C_d(K)$. Let $(T,X)$ be a nice tree decomposition of $K$. We will find $h$ using a dynamic program on $T$.
\par 
Any chain $c\in C_d(K)$ defines a chain homologous to $b$, namely $h=\boundary c+b$. Instead of searching for the $(d{-}1)$-chain $h$, we can therefore search for the $d$-chain $c$ that minimizes the weight $\|\boundary c+b\|$. The chain $c$ also defines the minimum chain homologous to $b$ in a local sense. If we restrict this chain to a subcomplex $K_t$, we find that $c\cap K_t$ is the chain that minimizes $\|\boundary (c\cap K_t)\setminus K[X_t] +b_t\|$ over all chains with boundary satisfying $\boundary (c\cap K_t)\cap K[X_t]=\beta$. Accordingly, for each $(d{-}1)$-chain $\beta\in C_{d{-}1}(K[X_t])$, we store the dynamic programming table entry
$$
\mathbf{V}[\beta,t]=\underset{\underset{\boundary\gamma\cap K[X_t]=\beta}{\gamma\in C_{d}(K_t):}}{\min} \|\boundary\gamma\setminus K[X_t]+b_t\|.
$$
\par\noindent
We say a $d$-chain $\gamma$ \textbf{\textit{spans $\boldsymbol\beta$ at t}} if $\boundary\gamma\cap K[X_t]=\beta$. Lemma \ref{lem:ohcp_opt} shows that the table $\V$ contains the weight of the optimal solution.
\begin{lemma}
\label{lem:ohcp_opt}
    The entry $\V[\emptyset, r]$ is the minimum weight $\|\boundary c+b\|$ for any chain $c\in C_d(K)$.
\end{lemma}
\begin{proof}
    The unique entry at the root $\mathbf{V}[\emptyset,r]$ contains the weight of the minimum weight $(d{-}1)$-chain homologous to $b$. The bag at the root $X_t$ is empty, so $K[X_r]=\emptyset$. Moreover, $V_r=V$, so $K_r=K$ and $b_r=b$. Thus, $\mathbf{V}[\emptyset,r]$ contains the minimum weight $\|\boundary c\setminus K[X_r]+b\|=\|\boundary c+b\|$ for any chain $c\in C_d(K_r)=C_d(K)$.
\end{proof}
\par
Our algorithm for OHCP is similar to our algorithm for OBCP as we are searching for a $d$-chain $c$, except we are optimizing a different function. We perform a dynamic program on our tree decomposition, and at each node $t$, we store the weight of the chain $\gamma\in C_d(K_t)$ that spans each $(d{-}1)$-chain $\beta\in C_{d{-}1}(K[X_t])$ and minimizes the objective function. We construct the chain $\gamma$ using the orthogonal decompositions of the chain groups $C_d(K_t)$ from Section \ref{section:mbc}. At each node, the two algorithms will loop over the same chain groups. Accordingly, the running time for our algorithm for OHCP is the same as the running time for our algorithm for OBCP, and we obtain Theorem \ref{thm:ohc}. 
\par
We now present our dynamic program for OHCP. We compute the entries of $\V$ at a node $t$ using the entries of $\V$ at the children of $t$. Accordingly, we only need to specify how to compute $t$ at each type of node in a nice tree decomposition.

\subsection{Leaf Nodes}
Let $t$ be a leaf node. Recall that $X_t=\emptyset$, so $K[X_t]=\emptyset$. Moreover, $t$ has no children, so $V_t=K_t=b_t=\emptyset$. There is a unique $(d{-}1)$-chain $\beta=\emptyset\in C_{d{-}1}(K[X_t])$ and a unique $d$-chain $\gamma=\emptyset\in C_d(K_t)$. The chain $\gamma$ spans $\beta$ at $t$, so $\mathbf{V}[\emptyset, t]=0$.

\subsection{Introduce Nodes}
Let $t$ be an introduce node and $t'$ the unique child of $t$. Recall that $X_t=X_{t'}\sqcup\{w\}$. Let $\beta\in C_{d{-}1}(K[X_t])$. We claim the following formula for $\mathbf{V}[\beta,t]$.
\begin{lemma}
Let $t$ be an introduce node and $t'$ be its unique child. Let $\beta\in C_{d-1}(K[X_t])$. Then
$$
\mathbf{V}[\beta,t]=
    \begin{cases}
        \mathbf{V}[\beta,t']&\text{ if }\beta\in C_{d{-}1}(K[X_{t'}])\\
        \infty & \text{ otherwise}
    \end{cases}
$$
\end{lemma}
\begin{proof}
    We will show that any chain that spans $\beta$ at $t'$ spans $\beta$ at $t'$ and vice versa.  Let $\gamma\in C_{d}(K_{t'})$ such that $\gamma$ spans $\beta$ at $t'$; that is, $\boundary\gamma\cap K[X_{t'}] = \beta$. We need to show that $\boundary\gamma\cap K[X_t]=\beta=\boundary\gamma\cap K[X_{t'}]$ as well. We do this by proving the equivalent statement that $\boundary\gamma\setminus K[X_t] = \boundary\gamma\setminus K[X_{t'}].$ As $\boundary\gamma\in K_{t}\cap K_{t'}$, then 
    \begin{align*}
        \boundary\gamma\setminus K[X_{t'}] =& \boundary\gamma\cap (K_{t'}\setminus K[X_{t'}]) \tag{as $\boundary\gamma\subset K_{t'}$} \\ 
        =& \boundary\gamma\cap(K_{t}\setminus K[X_t]) \tag{as $K_t\setminus K[X_t] = K_{t'}\setminus K[X_{t'}]$ by Lemma \ref{lem:mbc_introduce_complex}} \\ 
        =& \boundary\gamma\setminus K[X_t]. \tag{as $\boundary\gamma\subset K_{t}$}
    \end{align*}
    As a corollary, we see that $\boundary\gamma\cap K[X_t]=\boundary\gamma\cap K[X_{t'}]=\beta$. So if $\gamma$ spans $\beta$ at $t$, then $\gamma$ spans $\beta$ at $t'$ and vice versa. Moreover, as $b_t=b_{t'}$ by Lemma \ref{lem:mbc_introduce_boundary}, then $\|\boundary\gamma\cap K[X_t]+b_t\|=\|\boundary\gamma\cap K[X_{t'}]+b_{t'}\|$. So $\mathbf{V}[\beta,t]=\mathbf{V}[\beta,t]$ for each $(d{-}1)$-chain $\beta$ in $C_{d{-}1}(K_{t'})$. 
    \par
    As $\boundary\gamma\cap K[X_t]=\boundary\gamma\cap K[X_{t'}]$ for each $\gamma\in C_{d}(K_t)$, then if $\beta\in C_d(K[X_t])\setminus C_{d}(K[X_{t'}])$, no chain spans $\beta$ at $t$. Accordingly, we set $\mathbf{V}[\beta,t]=\infty$.
\end{proof}

\subsection{Forget Nodes}
Let $t$ be a forget node and $t'$ the unique child of $t$. Recall that $X_t\sqcup\{w\}=X_{t'}$. Let $\gamma\in C_d(K_t)$. We can decompose $\gamma=\gamma'+\gamma_w$ where $\gamma'\in C_d(K_{t'})$ and $\gamma_w\in C_{d}^{w}(K[X_t])$. Let $\beta'=\boundary\gamma'\cap K[X_{t'}]$. 
\par 
We want to find $\gamma$ that minimizes the weight of the chain $\boundary\gamma\setminus K[X_t]+b_t$. We can decompose the chain $\boundary\gamma\setminus K[X_t]+b_t$ into $(\boundary \gamma\setminus K[X_t]+b_t)\cap K[X_{t'}]$ and $(\boundary\gamma\setminus K[X_{t}]+b_t)\setminus K[X_{t'}]$. We find that $(\boundary\gamma\setminus K[X_t]+b_t)\cap K[X_t]=(\beta'+\boundary\gamma_w)\setminus K[X_{t}]+b_t\cap K[X_{t'}]$ and $(\boundary\gamma\setminus K[X_{t}]+b_t)\setminus K[X_{t'}]=\boundary\gamma'\setminus K[X_{t'}]+b_{t'}$. Decomposing $\boundary\gamma\setminus K[X_t]+b_t$ this way gives us a quick way to calculate $\| \boundary\gamma\setminus K[X_t]+b_t \|$. The chain $(\beta'+\boundary\gamma_w)\setminus K[X_{t}]+b_t\cap K[X_{t'}]$ is contained entirely in $C_{d-1}(K[X_{t'}])$, so its weight can be calculated in $O(k^d)$ time. Likewise, if $\gamma'$ is the chain that minimizes $\boundary\gamma'\setminus K[X_{t'}] + b_{t'}$ (which we will prove is the case for $\gamma$ that minimizes $\|\boundary\gamma\setminus K[X_t]+b_t \|$,) then the weight of $(\boundary\gamma' + b_{t'})\setminus K[X_{t'}]$ has already been computed and is stored in $\V[\beta',t']$. This justifies the following formula for $\V[\beta, t]$.
\begin{lemma}
Let $t$ be a forget node and $t'$ the unique child of $t$. Let $\beta\in C_{d-1}(K[X_t])$. Then
$$
\label{eqn:ohc_forget}
V[\beta,t]=
\underset{\beta',\gamma_w}
{\min} \; \V[\beta',t']+\|(\beta'+\boundary\gamma_w)\setminus K[X_t]+b_t\cap K[X_{t'}]\|
$$
where the minimization ranges over $\beta'$ and $\gamma_w$ such that 
\begin{itemize}
    \item $\beta'\in C_{d{-}1}(K[X_{t'}]$
    \item $\gamma_w\in C_{d}^{w}(K[X_{t'}])$
    \item $(\beta'+\boundary \gamma_w)\cap K[X_t]=\beta$
\end{itemize}
\end{lemma}
\begin{proof}
    We first show that any chain on the right hand side of the equation in the lemma spans $\beta$ at $t$.
    Let $\gamma'\in C_{d}(K_{t'})$ that achieves $\mathbf{V}[\beta',t]$ and let $\gamma_w\in C_d^{w}(K[X_{t'}])$ such that $(\beta'+\boundary\gamma_w)\cap K[X_t]=\beta$. Let $\gamma=\gamma'+\gamma_w$. Observe that $\boundary\gamma\cap K[X_t]\cap K[X_{t'}]=\boundary\gamma\cap K[X_t]$ as $K[X_t]\subset K[X_{t'}]$. We find that
    \begin{align*}
    \boundary\gamma\cap K[X_t] =& \boundary\gamma\cap K[X_{t'}]\cap K[X_t] \\ 
    =& (\boundary\gamma'\cap K[X_{t'}]+\boundary\gamma_w\cap K[X_{t'}])\cap K[X_t] \tag{as $\boundary\gamma = \boundary\gamma' + \boundary\gamma_w$}\\
    =& (\beta'+\boundary\gamma_w\cap K[X_{t'}])\cap K[X_t] \tag{as $\gamma'$ spans $\beta'$ at $t'$} \\
    =& (\beta'+\boundary\gamma_w)\cap K[X_t] \tag{as $\gamma'\subset K[X_{t'}]$} \\
    =& \beta \tag{by assumption}
    \end{align*}
     This proves that $\gamma$ indeed spans $\beta$ at $t$. This proof also works in the opposite direction. If $\gamma$ spans $\beta$ at $t$, and $\gamma$ decomposes $\gamma = \gamma' +\gamma_w$ for $\gamma'\in C_{d}(K_{t'})$, $\gamma_w \in C_d^w(K[X_{t'}])$, and $\boundary\gamma'\cap K[X_{t'}]=\beta'$, then this prove that $(\beta' + \boundary\gamma_w)\cap K[X_t] = \beta$.
    \par
    Next, we show that our algorithm calculates the weight of $\gamma$ correctly; that is, $\| \boundary\gamma\setminus K[X_t]+b_t \|=\mathbf{V}[\beta',t']+\|(\beta'+\boundary\gamma_w)\setminus K[X_t]+b_t\cap K[X_{t'}]\|$. As the table entry is defined $\V[\beta', t'] = \| \boundary\gamma'\setminus K[X_{t'}]+b_t  \|$, we can show the formula is correct by showing that $\boundary\gamma\setminus K[X_t] + b_t$ decomposes into
    $$
    \boundary\gamma\setminus K[X_t]+b_t=(\boundary\gamma'\setminus K[X_{t'}]+b_t)\sqcup ((\beta'+\boundary\gamma_w)\setminus K[X_t]+b_t\cap K[X_{t'}]).
    $$ 
    As a first step, we decompose $\boundary\gamma\setminus K[X_t]+b_t$ into the two disjoint sets $(\boundary\gamma\setminus K[X_t]+b_t)\cap K[X_{t'}]$ and $(\boundary\gamma\setminus K[X_t]+b_t)\setminus K[X_{t'}]$. In fact, these two sets will equal $(\beta'+\boundary\gamma_w)\setminus K[X_t]+b_t\cap K[X_{t'}]$ and $\boundary\gamma'\setminus K[X_{t'}]+b_{t'}$ respectively. If we intersect $\boundary\gamma\setminus K[X_t]+b_t$ with $K[X_{t'}]$, we find that
    \begin{align*}
        (\boundary\gamma\setminus K[X_t]+b_t)\cap K[X_{t'}] =& ((\boundary\gamma\setminus K[X_t])\cap K[X_{t'}])+(b_t\cap K[X_{t'}]) \\
        =& ((\boundary\gamma\cap K[X_{t'}])\setminus K[X_{t}])+(b_t\cap K[X_{t'}]) \tag{reordering} \\
        =& (\beta'+\boundary\gamma_w)\setminus K[X_t]+b_t\cap K[X_{t'}] \tag{*}
    \end{align*}
    where line (*) uses the fact $\boundary\gamma\cap K[X_{t'}]=\beta'+\boundary\gamma_w$ that was proved in the first paragraph of this proof. If we subtract $K[X_{t'}]$ from $\boundary\gamma\setminus K[X_{t'}]+b_t$, we find
    \begin{align*}
        (\boundary\gamma\setminus K[X_{t}]+b_t)\setminus K[X_{t'}] =& \boundary\gamma\setminus K[X_t]\setminus K[X_{t'}]+b_t\setminus K[X_{t'}] \\
        =& \boundary\gamma\setminus K[X_{t'}]+b_t\setminus K[X_{t'}] \tag{as $K[X_t]\subset K[X_{t'}]$}\\
        =& (\boundary\gamma'+\boundary\gamma_w)\setminus K[X_{t'}] +b_t\setminus K[X_{t'}] \tag{as $\boundary\gamma = \boundary\gamma' + \boundary\gamma_w$} \\
        =& \boundary\gamma'\setminus K[X_{t'}]+b_{t}\setminus K[X_{t'}] \tag{as $\boundary\gamma_w\subset K[X_{t'}]$}\\
        =& \boundary\gamma'\setminus K[X_{t'}]+b_{t'} \tag{as $b_t\setminus K[X_{t'}]=b_{t'}$ by Lemma \ref{lem:mbc_boundary_forget}}
    \end{align*}
    \par 
     Finally, we prove that the chain $\gamma$ that attains $\V[\beta, t]$ is included on the right hand side of the equation. We saw in the first paragraph that we can write $\gamma$ as $\gamma = \gamma' + \gamma_w$ where $\gamma' \in C_d(K_{t'})$ and $\gamma_w \in C_d^w(K[X_t])$. If $\boundary\gamma' \cap K[X_{t'}] = \beta'$, it is easy to see that $\gamma'$ must achieve $\V[\beta',t']$; otherwise, we could substitute $\gamma'$ for the chain that achieves $\V[\beta',t']$ and get a strictly smaller chain.
\end{proof}

\subsection{Join Nodes}

Let $t$ be a join node, and let $t'$ and $t''$ be the two children of $t$. Recall that $X_t=X_{t'}=X_{t''}$. Let $\beta\in C_{d{-}1}(K[X_t])$. We claim the following formula for $\mathbf{V}[\beta,t]$.

\begin{lemma}
\label{eqn:ohc-join}
Let $t$ be a join node, and let $t'$ and $t''$ be its two children. Let $\beta\in C_{d{-}1}(K[X_t])$. Then
$$
\mathbf{V}[\beta,t]=\underset{\beta', \beta''}{\min}
\mathbf{V}[\beta',t']+\mathbf{V}[\beta'',t'']
$$
where the minimization ranges over $\beta'$ and $\beta''$ such that
\begin{itemize}
    \item $\beta'\in C_{d{-}1}(K[X_{t'}])$
    \item $\beta''\in C_{d{-}1}(K[X_{t''}])$
    \item $\beta'+ \beta''=\beta$
\end{itemize}
\end{lemma}

\begin{proof}
Let $\beta\in C_{d{-}1}(K_{t})$, and let $\gamma$ span $\beta$ at $t$. By Lemma \ref{lem:mbc_join_chain}, $\gamma=\gamma'+\gamma''$ where $\gamma'\in C_d(
K_{t'})$ and $\gamma''\in C_d(K_{t''})$. Let $\beta'=\boundary\gamma'\cap K[X_{t'}]$ and $\beta''=\boundary\gamma''\cap K[X_{t''}]$. Clearly, 
\begin{align*}
\beta =& \boundary\gamma\cap K[X_t] \\
=& (\boundary\gamma'+\boundary\gamma'')\cap K[X_t] \\
=& \boundary\gamma'\cap K[X_{t}]+\boundary\gamma''\cap K[X_{t}] \\
=& \boundary\gamma'\cap K[X_{t'}]+\boundary\gamma''\cap K[X_{t''}] \tag{as $K[X_t] = K[X_{t'}] = K[X_{t''}]$}\\
=& \beta'+\beta''.
\end{align*}
This proof works in the reverse direction as well. If $\gamma'$ spans $\beta'$ and $t'$, $\gamma''$ spans $\beta''$ at $t''$, and $\beta'+\beta''=\beta$, then $\gamma'+\gamma''$ spans $\beta$ at $t$.
\par
Next, we prove that the formula correctly calculates the weight of $\gamma$. We first prove that $\boundary\gamma\setminus K[X_t]+b_t=(\boundary\gamma'\setminus K[X_{t'}]+b_{t'})\sqcup (\boundary\gamma''\setminus K[X_{t''}]+b_{t''})$. We consider the intersection 
\begin{align*}
    \boundary\gamma\setminus K[X_t] =& \boundary\gamma\cap K_t\setminus K[X_t] \\
    =& (\boundary\gamma' + \boundary\gamma'')\cap K_t\setminus K[X_t] \tag{as $\boundary\gamma = \boundary\gamma' + \boundary\gamma''$} \\
    =& (\boundary\gamma'+\boundary\gamma'')\cap(K_{t'}\setminus K[X_{t'}]+K_{t''}\setminus K[X_{t''}]) \tag{by Lemma \ref{lem:mbc_join_complex}}
\end{align*}
As $\boundary\gamma'\subset K_{t'}$, $\boundary\gamma''\subset K_{t''}$, and $K_{t'}\setminus K[X_{t'}]$ and $K_{t''}\setminus K[X_{t''}]$ are disjoint, we distribute and see that
$$
\boundary\gamma\setminus K[X_t]=\boundary\gamma'\setminus K[X_{t'}]+\boundary\gamma''\setminus K[X_{t''}]
$$
Now if we consider the intersection
\begin{align*} 
    \boundary\gamma\setminus K[X_t]+b_t =& (\boundary\gamma'\setminus K[X_{t'}]+\boundary\gamma''\setminus K[X_{t''}])+(b_{t'}+b_{t''}) \tag{by Corollary \ref{cor:mbc_join_boundary} }\\
    =& (\boundary\gamma'\setminus K[X_{t'}]+b_{t'})\sqcup(\boundary\gamma'\setminus K[X_{t''}]+b_{t''})
\end{align*}
where the second line follows as $\boundary\gamma\setminus K[X_{t'}]+b_{t'}\subset K_{t'}\setminus K[X_{t'}]$, $\boundary\gamma\setminus K[X_{t''}]+b_{t''}\subset K_{t''}\setminus K[X_{t''}]$, and $K_{t'}\setminus K[X_{t'}]$ and $K_{t''} \setminus K[X_{t''}]$ are disjoint. Thus, $\|\boundary\gamma\setminus K[X_t]+b_t\|=\|\boundary\gamma'\setminus K[X_{t'}]+b_{t'}\|+\|\boundary\gamma''\setminus K[X_{t''}]+b_{t''}\|$.
\par
As the sum of any two chains $\gamma'$ and $\gamma''$ that span $\beta'$ at $t'$ and $\beta''$ at $t''$ spans $\beta$ at $t$, it is easy to see that the optimal chain that span $\beta$ at $t$ must be composed of optimal chains that span $\beta'$ at $t'$ and $\beta''$ at $t''$ respectively.
\end{proof}

\section{Comparison with Other Algorithms}
\label{sec:comparison}
\newcommand{\con}{Con_d(K)}
\newcommand{\hasse}{Hasse_d}

In this section, we compare our algorithms for OHCP and OBCP to the algorithms introduced by Blaser and V\.agset \cite{blaser_hl} and Blaser et al.~\cite{blaser_mbc}. Our algorithm and their algorithms are parameterized by the treewidth of different graphs associated with a simplicial complex, so we compare the treewidth of the 1-skeleton to the treewidth of these graphs. We start by defining these graphs.
\par 
The \textit{\textbf{d-connectivity graph}} is the graph $Con_d(K)$ with the $d$-simplices $K_d$ as vertices and pairs of $d$-simplices that share a $(d{-}1)$-face as edges. \textit{\textbf{Level d of the Hasse diagram}} is the graph $\hasse(K)$ with the $(d{-}1)$ and $d$-simplices $K_{d-1}\cup K_d$ as vertices and edges between $d$-simplices and their $(d{-}1)$-faces.
\par 
Our algorithms for OHCP and OBCP run in $2^{O(k^d)}n$ time, where $k$ is the treewidth of the 1-skeleton and $n$ is the number of vertices in the complex. Blaser and V\.agset introduce two algorithms for OHCP: one algorithm uses the $d$-connectivity graph, and the other algorithm uses level $d$ of the Hasse diagram of $K$. The algorithms of Blaser and V\.{a}gset run in $O(2^{l(2d+5)}m)$ and $O(2^{2l}m)$ time respectively, where $l$ is the treewidth and $m$ is the number of vertices in the given graph. Blaser et al.~introduce an algorithm for OBCP that runs in $O(4^{l}m)$ time where $l$ is the treewidth and $m$ is the number of vertices of the Hasse diagram. Note that for the $d$-connectivity graph and level $d$ of the Hasse diagram, the number of vertices $m$ are $O(\binom{n}{d+1})$ and $O(\binom{n}{d}+\binom{n}{d+1})$ respectively, and these bounds can be tight.
\par 
The goal of this section is to compare the treewidth of the graphs used by our algorithm and by the algorithms of Blaser and V\.agset and Blaser et al. The following theorem by Blaser and V\.agset compare the treewidth of $\con$ and $\hasse(K)$

\begin{theorem}[Proposition 7.1, Blaser and V\.agset~\cite{blaser_hl}]
    Let $K$ be a simplicial complex. The treewidth $\tw(\hasse(K))\leq \tw(\con)+1$; however, $\tw(\con)$ is not bounded by $\tw(\hasse(K))$.
\end{theorem}

As the treewidth of $\hasse(K)$ is $O(\tw(\con))$, we will spend the rest of this section comparing the treewidth of $\hasse(K)$ and the treewidth of the 1-skeleton $K^1$. We will show that $\tw(\hasse(K))\in O(\binom{\tw(K^1)+1}{d})$. We will then show that for a certain class of complexes, this upper bound is nearly tight, i.e. $\tw(\hasse(K))\in \Omega(\binom{\tw(K^1)+1}{d}\cdot \frac{1}{d+1})$. The first result implies that algorithms parameterized by $\tw(\hasse(K))$ having running time $O(2^{\tw(K^1)^{d}}m)$; in other words, these algorithm have (at worst) the same dependence on $\tw(K^1)$ as our algorithms. The second result implies that in some cases, the dependence on $\tw(K^1)$ is almost the same.

\subsection{Upper Bounds}

In this section, we provide an upper bound on the treewidth of level $d$ of the Hasse diagram in terms of the treewidth of the 1-skeleton. The proof of this theorem was developed with Nello Blaser and Erlend Raa V\.{a}gset.

\begin{theorem}
\label{thm:tw_hasse_upper_bound}
    Let $K$ be a simplicial complex. Then $\tw(\hasse(K))\in O(\binom{\tw(K^1)+1}{d})$.
\end{theorem}
\begin{proof}
     Let $(T,X)$ be any tree decomposition of $K^1$ of width $k$. We will use $(T,X)$ as a starting point for constructing a tree decomposition of $\hasse(K)$ of width $O(\binom{k+1}{d})$. For a node $t\in T$, we add all $(d-1)$-simplices in $K[X_t]$ to a new bag $Y_t$. The size of any bag $Y_t$ is $O(\binom{k+1}{d})$, as any set of $d$ vertices in $X_t$ may be a $(d-1)$-simplex in $K$. By Corollary \ref{cor:simplex_bag}, each $(d-1)$-simplex $\sigma\in Y_t$ for some $t\in T$ as $\sigma\subset X_t$ for some $t\in T$, and moreover, the set of nodes $\{t\in T:\sigma\in T\}$ are connected.
     \par 
     Next, for each $d$-simplex $\tau\in K$, there is a node $t\in T$ such that $\tau\subset X_t$ by Corollary \ref{cor:simplex_bag}. By construction, we conclude that each $(d-1)$-face $\sigma$ of $\tau$ must be contained in the new bag $Y_t$. We add a new node $t_\tau$ to the tree with bag $Y_{t_\tau} = Y_t\cup\{\tau\}$, and we connect $t_\tau$ to $t$. The size of $Y_{t_\tau}$ is $O(\binom{k+1}{d})$, as $Y_{t_\tau}$ contains one more vertex than $Y_t$.
     \par 
     We now verify that this new tree decomposition satisfies the conditions to be a tree decomposition. Each $(d-1)$-simplex $\sigma$ is contained in the bags of the nodes $\{t\in T:\sigma\subset X_t\}\cup\{t_\tau:\sigma\subset\tau\}$; the nodes $\{t\in T:\sigma\subset X_t\}$ form a connected subtree as $(T,X)$ is a tree decomposition, and each node $t_\tau$ is connected to some node in $\{t\in T:\sigma\subset X_t\}$. Each $d$-simplex $\tau$ is contained in the single bag $Y_{t_\tau}$. Finally, each edge in $\hasse(K)$ is a pair $\{\sigma,\tau\}$ of a $(d-1)$-simplex $\sigma$ and a $d$-simplex $\tau$ such that $\sigma\subset\tau$. By construction, $\{\sigma,\tau\}\subset Y_{t_\tau}$.
     \par 
     This proves there is a tree decomposition of $\hasse(K)$ of width $O(\binom{k+1}{d})$ when $k$ is the width of any tree decomposition of $K^1$. In particular, this proves that there is a tree decomposition of width $O(\binom{\tw(K^1)+1}{d})$.
\end{proof}

\subsection{Lower Bounds}

Let $\Delta^n$ be the simplicial complex that is the set of all subsets of a set of vertices $V=\{v_1,\ldots,v_n\}$. (Intuitively, $\Delta^n$ is a $(n-1)$-simplex and all its faces.) Our main theorem of this section is a lower bound of $\Omega(\binom{n}{d}\cdot\frac{1}{d+1})$ on $\tw(\hasse(\Delta^n))$. The treewidth of the 1-skeleton $\tw((\Delta^n)^1) = n-1$ as $(\Delta^n)^1$ is a complete graph, so this lower bound shows that the $\Delta^n$ is close to being a worst-case for the treewidth of the Hasse diagram. Specifically, the lower bound on $\tw(\hasse(\Delta^n))$ differs from the upper bound on $\tw(\hasse(\Delta^n))$ of Theorem \ref{thm:tw_hasse_upper_bound} by a multiplicative factor of $d+1$. 

\subsubsection{Background: Vertex and Edge Expansion}

We will obtain our lower bounds of $\tw(\hasse(\Delta^n))$ using the related notions of vertex and edge expansion. Let $G$ be a graph. The \textit{\textbf{vertex expansion}} of $G$ is 
$$
    VE(G) = \underset{S\subset V(G) : 1\leq |S| \leq \frac{1}{2}|V(G)|}{\min} \frac{|N(S)|}{|S|}
$$
where $N(S)$ is the set of vertices in $V(G)\setminus S$ with a neighbor in $S$. The vertex expansion of $G$ is closely tied to the treewidth of $G$, as evidenced by the following lemma from Chandran and Subramian \cite{chandran_subramian+2005}.\par

\begin{lemma}[Chandran and Subramian, Lemma 9\cite{chandran_subramian+2005}]
\label{lem:tw_lower_bound}
    Let $1\leq s\leq |V(G)|$. Define $N_{\min}(G,s)$ to     be
    $$
        N_{\min}(G,s) = \underset{S\subset V(G) : s/2 \leq |S| \leq s}{\min} |N(S)|.
    $$
    Then $\tw(G) \geq N_{\min}(G,s)-1$.
\end{lemma}
\begin{corollary}
\label{cor:tw_lower_bound_ve}
    Let $G$ be a graph. Then $\tw(G) \geq \frac{VE(G)\cdot |V(G)|}{4}$.
\end{corollary}
\begin{proof}[Proof of Corollary]
    We obtain the corollary by setting $s = |V(G)| / 2$. Using the definition of vertex expansion, any set of vertices $S$ with at least $s/2 = |V(G)|/4$ vertices satisfies $|N(S)| \geq VE(G)\cdot |S|\geq \frac{VE(G)\cdot |V(G)|}{4}$.
\end{proof}

We are able to get a lower bound on the vertex expansion (and by Corollary \ref{cor:tw_lower_bound_ve}, a lower bound on the treewidth of $G$) using the related notion of edge expansion. The \textit{\textbf{edge expansion}} of a graph $G$ is 
$$
    EE(G) = \underset{S\subset V(G) : 1\leq |S| \leq |V(G)|/2}{\min} \frac{|\delta(S)|}{|S|}
$$
where $\delta(S)$ is the set of edges with one endpoint in $S$ and one endpoint in $V(G)\setminus S$. The following lemma relates the notion of edge and vertex expansion. 
\begin{lemma}
\label{lem:ve_vs_ee}
    Let $d_{\max}$ be the maximum degree of a vertex in $G$. Then $VE(G) \geq \frac{EE(G)}{d_{\max}}$. 
\end{lemma}
\begin{proof}
     Let $S\subset V(G)$ be a set of vertices such that $1\leq |S| \leq |V(G)|/2$. Then $|N(S)| \geq |\delta(S)| / d_{\max}$ as each vertex in $N(S)$ is incident to at most $d_{\max}$ edges in $\delta(S)$. If $S^* = \underset{S\subset V(G) : 1\leq |S|\leq |V(G)|/2}{\arg\min} \frac{|N(S)|}{|S|}$, then 
     $$
        VE(G) = \frac{|N(S^*)|}{|S^*|} \geq \frac{|\delta(S^*)|}{d_{\max}\cdot|S^*|} \geq \frac{EE(G)}{d_{\max}}
     $$
\end{proof}

We can now combine Corollary \ref{cor:tw_lower_bound_ve} and Lemma \ref{lem:ve_vs_ee} to get a lower bound on the treewidth in terms of edge expansion.

\begin{lemma}
\label{cor:tw_lower_bound_ee}
    Let $G$ be a graph. Then $\tw(G) \geq \frac{EE(G)\cdot |V(G)|}{4\cdot d_{\max}}$.
\end{lemma}

\subsubsection{Background: Edge Transitivity}

We are able to get a lower bound on the edge expansion of $\hasse(\Delta^n)$ as $\hasse(\Delta^n)$ is \textit{edge-transitive}. Let $G$ be a graph. A \textit{\textbf{graph automorphism}} is a bijection $\phi:V(G)\to V(G)$ such that $\{u,v\}\in E(G)$ if and only if $\{\phi(u),\phi(v)\}\in E(G)$. A graph $G$ is \textit{\textbf{edge-transitive}} if for any two edges $\{u_1,v_1\}$ and $\{u_2,v_2\}$ of $G$, there is an automorphism $\phi:V(G)\to V(G)$ such that $\{\phi(u_1),\phi(v_1)\} = \{u_2,v_2\}$.

\begin{lemma}
\label{lemma:edge_transitive}
    The graph $\hasse(\Delta^n)$ is edge-transitive.
\end{lemma}
\begin{proof}
    Recall that $\Delta^n$ is the set of all subsets of a set $V=\{v_1,\ldots,v_n\}$, and $\hasse(\Delta^n)$ has vertices that are the set of all subsets of size $d$ and $d+1$ of $V$. As a stepping stone to the lemma, we claim that any bijection of $f:V\to V$ defines a graph automorphism $\phi_{f}:V(\hasse(\Delta^n))\to V(\hasse(\Delta^n))$. Specifically, if $\sigma=\{u_1,\ldots,u_d\}\in V(\hasse(\Delta^n))$ is a $(d-1)$-simplex, we define $\phi_f$ to map $\phi_f(\sigma) = \{f(u_1),\ldots,f(u_d)\}$. Note that $\phi_f(\sigma)$ is a set with $d$ distinct elements as $f$ is a bijection; or in other words, $\phi_f(\sigma)$ is a $(d-1)$-simplex. We define $\phi_f$ analogously on $d$-simplices in $V(\hasse(\Delta^n))$. 
    \par 
    We now must verify that $\phi_f$ is a graph automorphism. Any edge $\{\sigma,\tau\}$ of $\hasse(\Delta^n)$ is a pair of a $(d-1)$-simplex $\sigma$ and a $d$-simplex $\tau$ such that $\sigma\subset\tau$. It is straightforward to verify that if $\{\sigma,\tau\}$ is an edge of $\hasse(\Delta^n)$, then $\{\phi_f(\sigma),\phi_f(\tau)\}$ is also an edge; indeed, $\phi_f(\sigma)$ and $\phi_f(\tau)$ are a $(d-1)$ and $d$-simplex respectively, and $\phi_f(\sigma)\subset\phi_f(\tau)$ by definition. It is also straightforward to verify the converse.
    \par 
    We now use this fact to show that $\hasse(\Delta^n)$ is edge-transitive. Let $\{\sigma_1,\tau_1\}$ and $\{\sigma_2,\tau_2\}$ be any two edges of $\hasse(\Delta^n)$ such that $\sigma_1=\{u_0,\ldots,u_{d-1}\}$, $\tau_1 = \{u_0,\ldots,u_{d}\}$, $\sigma_2 = \{v_0,\ldots,v_{d-1}\}$, and $\tau_2 = \{v_0,\ldots,v_d\}$. Let $f:V\to V$ be any bijection such that $f(u_i)=v_i$ for $0\leq i\leq d$. From the previous argument, the map $f$ defines a graph automorphism $\phi_f$ on $\hasse(\Delta^n)$. It follows immediately from the definition of $\phi_f$ that $\phi_f(\sigma_1) = \sigma_2$ and $\phi_f(\tau_1) = \tau_2$.
\end{proof}
The following lemma by Babai and Szegedy \cite{babai_szegedy_1992} provides a lower bound on the edge expansion of edge-transitive graphs. Let $a$ and $b$ be non-zero real numbers. The \textit{\textbf{harmonic mean}} of $a$ and $b$ is $2\cdot\left(\frac{1}{a} + \frac{1}{b}\right)^{-1}$. 

\begin{lemma}[Babai and Szegedy, 1992]
\label{lem:expansion_lower_bound}
    Let $G$ be a simple, edge-transitive graph. Let $D$ be the diameter of $G$, and let $r$ be the harmonic mean of the minimum and maximum degree of $G$. The edge expansion of $G$ is $EE(G) \geq \frac{r}{2D}$. 
\end{lemma}

\subsubsection{Putting Everything Together}

We now calculate the diameter $D$ of $\hasse(\Delta^n)$ and the parameter $r$ from the statement of Lemma \ref{lem:expansion_lower_bound}.

\begin{lemma}
\label{lemma:diameter}
    The diameter $D$ of $\hasse(\Delta^n)$ is bounded above by $2d+2$.
\end{lemma}
\begin{proof}
    Let $\sigma$ and $\tau$ be distinct vertices of $\hasse(\Delta^n)$. We claim there is a $|\sigma\oplus\tau|$-length path between $\sigma$ and $\tau$ in $\hasse(\Delta^n)$, where $\sigma\oplus\tau$ is the symmetric difference of the sets $\sigma$ and $\tau$. To prove this, we will show there is a simplex $\sigma'$ adjacent to $\sigma$ such that $|\sigma'\oplus\tau| = |\sigma\oplus\tau|-1$. This will imply the lemma, as the largest symmetric difference between two simplices in $\hasse(\Delta^n)$ is between two disjoint $(d+1)$-simplices, where the symmetric difference is $2d+2$.  
    \par 
    We can prove the existence of $\sigma'$ by induction. First consider the case that $|\sigma\oplus\tau|=1$. Then we claim it must be the case that $\sigma\subset\tau$ or $\tau\subset\sigma$. Indeed, if $\sigma\not\subset\tau$ and $\tau\not\subset\sigma$, then $|\sigma\oplus\tau|\geq 2$. In this case, $\sigma'=\tau$.
    \par 
    In the case that $|\sigma\oplus\tau|>1$, there are two possibilities. If $\sigma$ is a $d$-simplex, then $\sigma' = \sigma\setminus\{v\}$ for any $v\in\sigma\setminus\tau$. If $\sigma$ is a $(d-1)$-simplex, then $\sigma' = \sigma\cup\{v\}$ for any $v\in\tau\setminus\sigma$. In both cases, the simplex $\sigma'$ is a vertex in $\hasse(\Delta^n)$ by the definition of $\hasse(\Delta^n)$, and $|\sigma'\oplus\tau| = |\sigma\oplus\tau|-1$.
\end{proof}

\begin{lemma}
\label{lem:harmonic_mean}
    Let $n \geq 2d+1$. The harmonic mean $r$ of the min and max degree of $\hasse(\Delta^n)$ of the $n$-simplex is bounded below by $d+1$.
\end{lemma}
\begin{proof}
    The degree of a vertex $\sigma\in V(\hasse(\Delta^n))$ is either $d+1$ and $n-d$ if $\sigma$ is a $d$ or $(d-1)$-simplex respectively, so $d+1$ and $n-d$ are the maximum and minimum degree of $\hasse(\Delta^n)$. We lower bound the harmonic mean as 
    \begin{align*}
        2\left(\frac{1}{d+1} + \frac{1}{n-d}\right)^{-1} &\geq 2\left(\frac{1}{d+1} + \frac{1}{d+1}\right)^{-1} \tag{as $n\geq 2d+1$} \\ 
        &= d+1
    \end{align*}
\end{proof}

We now have everything we need to lower bound the edge expansion of $\hasse(\Delta^n)$.

\begin{lemma}
    Let $n\geq 2d+1$. The edge expansion of $\hasse(\Delta^n)$ is bounded below by $\frac{1}{4}$.
\end{lemma}
\begin{proof}
    We can lower bound the edge expansion of $\hasse(\Delta^n)$ by plugging the values calculated in Lemma \ref{lemma:diameter} and Lemma \ref{lem:harmonic_mean} into the formula given by Lemma \ref{lem:expansion_lower_bound}, namely $EE(\hasse(\Delta^n)) \geq \frac{r}{2D} = \frac{d+1}{2\cdot(2d+2)} = \frac{1}{4}.$
\end{proof}
\begin{theorem}
\label{thm:lower_bound}
    Let $n\geq 2d+1$. The treewidth of $\hasse(\Delta^n)$ is $\Omega(\binom{n}{d}\cdot\frac{1}{d+1})$
\end{theorem}
\begin{proof}
     The number of vertices in $\hasse(\Delta^n)$ is $|V(\hasse(\Delta^n))| = \binom{n}{d} + \binom{n}{d+1}$. As $n \geq 2d+1$, the maximum degree of $\hasse(\Delta^n)$ is $d_{\max} = n-d$. We plug these values and the lower bound on $EE(\hasse(\Delta^n))$ into the equation given by Corollary \ref{cor:tw_lower_bound_ee} to obtain a lower bound on $\tw(\hasse(\Delta^n))$.
     \begin{align*}
         \tw(\hasse(\Delta^n)) &\geq \frac{EE(\hasse(\Delta^n))\cdot |V(\hasse(\Delta^n))|}{4\cdot d_{\max}} \\
         &\geq \frac{\binom{n}{d} + \binom{n}{d+1}}{4\cdot 4\cdot(n-d)} \\
         &\geq \frac{\binom{n}{d+1}}{4\cdot 4\cdot(n-d)} \\
     \end{align*}
     We can derive the equality $\binom{n}{d+1} \cdot \frac{1}{n-d} = \binom{n}{d}\cdot\frac{1}{d+1}$ by expanding $\binom{n}{d+1}$, which gives the theorem statement. 
\end{proof}

\section{Acknowledgements}

We would like to thank Nello Blaser and Erlend Raa V\.{a}gset for helpful discussions about Section \ref{sec:comparison}, specifically for helping us prove Theorem \ref{thm:tw_hasse_upper_bound}.

%
%

\bibliographystyle{plain}
\bibliography{main.bib}

\end{document}